\documentclass[11pt]{article}
\usepackage{bookmark}
\usepackage{multirow}
\usepackage{amsthm,amssymb,amsmath,mathrsfs,cite,cleveref,float,setspace}
\usepackage{enumerate}
\usepackage[shortlabels]{enumitem}
\usepackage[top=0.8 in, left=0.9 in, right=1.0 in, bottom=0.8 in]{geometry}
\usepackage{epstopdf}

\usepackage[colorinlistoftodos]{todonotes}
\usepackage[framemethod=TikZ]{mdframed}

\usepackage{lineno}

\newtheorem {theorem}{Theorem}
\newtheorem {lemma}{Lemma}
\newtheorem {proposition}{Proposition}
\newtheorem {corollary}{Corollary}

\usepackage{graphicx,caption,subcaption,epstopdf,enumitem,tikz,standalone}
\usetikzlibrary{shapes,backgrounds}
\usetikzlibrary[positioning,decorations.pathmorphing]

\newcommand{\known}[1]{{\color{blue} #1}}
\newcommand{\improved}[1]{{\color{red} #1}}
\newcommand{\hhjcss}[1]{{\color{black} #1}}

\theoremstyle{definition}
\newmdtheoremenv[
   innerleftmargin = 4pt,
   innerrightmargin = 4pt,
   innertopmargin = 4pt,
   innerbottommargin = 4pt,
   linecolor = black,
   middlelinewidth = 0.4pt,
   splitbottomskip = 4pt,
   splittopskip = 12pt,
]{openquestion}{Open Question}[section]

\newcommand{\strictMinp}{strict end-minimal}

\newcommand{\minp}{end-minimal}
\newcommand{\minpath}{end-minimal shortest path}
\newcommand{\aMinpath}{an end-minimal shortest path}

\usepackage{authblk}

\title{Eccentricity terrain of $\delta$-hyperbolic graphs}
\author{\medskip Feodor F. Dragan and Heather M. Guarnera   \\
Algorithmic Research Laboratory, Department of Computer Science \\
Kent State University, Kent, Ohio, USA \\
\href{mailto:dragan@cs.kent.edu}
{\texttt{dragan@cs.kent.edu}},
\href{mailto:hmichaud@kent.edu}
{\texttt{hmichaud@kent.edu}}
}

\begin{document}
\maketitle

\begin{abstract} A graph $G=(V,E)$ is $\delta$-hyperbolic if for any four vertices $u,v,w,x$, the two larger of the three distance sums $d(u,v)+d(w,x)$, $d(u,w)+d(v,x)$, $d(u,x)+d(v,w)$ differ by at most $2\delta\geq 0$. This paper describes the eccentricity terrain of a  $\delta$-hyperbolic graph.  The eccentricity function $e_G(v)=\max\{d(v,u):u\in V\}$ partitions vertices of $G$ into eccentricity layers $C_{k}(G)=\{v\in V:e_G(v)=rad(G)+k\}$, $k\in\mathbb{N}$, where $rad(G)=\min\{e_G(v):v\in V\}$  is the radius of $G$. The paper studies the eccentricity layers of vertices along shortest paths, identifying such terrain features as hills, plains, valleys, terraces, and plateaus. It introduces the notion of $\beta$-pseudoconvexity, which implies Gromov's $\epsilon$-quasiconvexity, and illustrates the abundance of pseudoconvex sets in $\delta$-hyperbolic graphs. It shows that all sets $C_{\leq k}(G)=\{v\in V:e_G(v)\leq rad(G)+k\}$, $k\in \mathbb{N}$, are $(2\delta-1)$-pseudoconvex.  Several bounds on the eccentricity of a vertex are obtained which yield a few approaches to efficiently approximating all eccentricities. 
\end{abstract}

\medskip 


\section{Introduction}
The \emph{eccentricity} $e_G(v)$ of a vertex~$v$ is the maximum distance from~$v$ to any other vertex in $G=(V,E)$, i.e., $e_G(v) = \max_{u
\in V}d(u,v)$.
The diameter $diam(G)$ (radius $rad(G)$) denotes the maximum (minimum) eccentricity of a vertex in~$G$.
The eccentricity function partitions the vertex set of~$G$ into eccentricity layers,
wherein each layer is defined as $C_{k}(G) = \{v \in V : e_G(v) = rad(G) + k\}$ for an integer $k \in [ 0, diam(G) - rad(G) ]$.
As the eccentricities of two neighboring vertices~$u$ and~$v$ can differ by at most one,
if vertex~$u$ belongs to layer~$C_k(G)$, then any vertex~$v$ adjacent to~$u$ belongs to either $C_{k-1}(G)$, $C_{k}(G)$, or $C_{k+1}(G)$.
The first layer $C_{0}(G)$ is exactly the center~$C(G)$ (all vertices of $G$ with minimum eccentricity
).
The last layer $C_{p}(G)$, where $p = {diam(G)-rad(G)}$, consists of all diametral vertices~$v$, i.e., with $e_G(v) = diam(G)$.
Also of interest are the sets defined as $C_{\leq k}(G) = \{v \in V : e_G(v) \leq rad(G) + k\}$,
that is, the union of all eccentricity layers from $C_0(G)$ to $C_k(G)$.
\hhjcss{The locality of a vertex~$v \notin C(G)$ is the minimum distance from~$v$ to a vertex with smaller eccentricity: $loc(v) = \min\{ d(v,x) : x \in V,\ e_G(x) < e_G(v)\}$; by definition, the locality of a central vertex is 0.}

\hhjcss{
The \emph{eccentricity terrain} illustrates the behavior of the eccentricity function along any shortest path: if a traveler begins at vertex~$y$ and ends at vertex~$x$ moving along $P(y,x)$,  
he may describe his journey as a combination of walking up-hill (to a vertex of higher eccentricity), down-hill (to a vertex of lower eccentricity), 
or along a plain (no change in eccentricity).
We identify such terrain features as hills, plains,  valleys, terraces, and plateaus.}
Understanding the eccentricity terrain and being able to efficiently estimate the diameter, radius, and all vertex eccentricities is of great importance.
For example, in the analysis of social networks
(e.g., citation networks or recommendation networks), biological systems (e.g., protein interaction networks),
computer networks (e.g., the Internet or peer-to-peer networks), transportation networks (e.g., public transportation
or road networks), etc., the {\it eccentricity} $e_G(v)$ of a vertex $v$ is used to measure the importance of $v$ in the network:
the {\em eccentricity centrality index} of $v$ \cite{Brandes} is defined as $\frac{1}{e_G(v)}$.

This paper further investigates the eccentricity function in $\delta$-hyperbolic graphs from the eccentricity terrain prospective and greatly advances the line of research taken in~\cite{Dragan2018,ourManuscriptDHG,Dragan2017EccentricityAT,DBLP:journals/dm/Prisner00a} for such special graph classes as 
chordal graphs, $(\alpha_1, \triangle)$-metric graphs and distance-hereditary graphs and in~\cite{Alrasheed_2016,Dragan2018RevisitingRD,Chepoi_2008,Chepoi2018FastAO} for general $\delta$-hyperbolic graphs and their relatives.
Gromov~\cite{Gromov1987} defines $\delta$-hyperbolic graphs via a simple 4-point condition:
for any four vertices $u,v,w,x$, the two larger of the three distance sums $d(u,v) + d(w,x)$, $d(u,w) + d(v,x)$, and $d(u,x) + d(v,w)$ differ by at most $2\delta \geq 0$.
Such graphs have become of recent interest due to the empirically established presence of a small hyperbolicity in many real-world networks, such as biological networks, social networks, Internet application networks, and collaboration networks, to name a few (see, e.g., \cite{AADr,AdcockSM13,Bo++,KeSN16,NaSa,ShavittT08}).
Notice that any graph is $\delta$-hyperbolic for some hyperbolicity $\delta\le diam(G)/2$.

\subsection{Our contribution}
 
 First, we define in Section~\ref{section:convexity} a $\beta$-pseudoconvexity which implies the quasiconvexity found by Gromov in hyperbolic graphs, but additionally, is closed under intersection. Interestingly, all disks and all sets $C_{\leq k}(G)$, for any integer $k \geq 0$, are $(2\delta - 1)$-pseudoconvex \hhjcss{in $\delta$-hyperbolic graphs.}

In Section~\ref{section:terrainShapes}, we show that the height of any up-hill as well as the width of any plain on a shortest path to a central vertex is small and depends (linearly) only on the hyperbolicity of $G$.
    Moreover, the cumulative height and width of all up-hills and plains on any shortest path to a vertex with minimal eccentricity is no more than $4\delta$.
    On any given shortest path $P$ from an arbitrary vertex to a closest central vertex, the number of vertices with locality more than 1  does not exceed $\max\{0,4\delta-1\}$. Furthermore, only at most $2\delta$ of them are  located outside $C_{\le \delta}(G)$ and only at most $2\delta+1$ of them are at distance $>2\delta$ from $C(G)$.  On the negative side, we give an example which illustrates that up-hills can occur anywhere on any shortest path from a vertex to a closest central vertex.

In Section~\ref{section:eccentricityBounds}, we give upper and lower bounds on the eccentricity of a vertex~$v$ based on several situations: if $v$ is on a shortest path $P(x,c)$ from a vertex $x$ to a closest central vertex $c$; 
if $v$ is on a shortest $(x,y)$-path where $y$ is a most distant vertex from $x$; 
if $v$ is on a shortest $(x,y)$-path where $x$ and $y$ are mutually distant vertices; and
if $v$ is a furthest vertex from some arbitrary vertex~$c \in V$.
Such results also give lower bounds on $diam(G)$ and upper bounds on $rad(G)$ which are consistent with those found in literature~\cite{Chepoi_2008,Dragan2018RevisitingRD}. More importantly, they are very useful in approximating all eccentricities in $G$. 

Finally, we present 
three approximation algorithms for all eccentricities:
an $O(\delta|E|)$ time eccentricity approximation $\hat{e}(v)$ based on the distances from any vertex to two mutually distant vertices which
satisfies $e_G(v) - 2\delta \leq \hat{e}(v) \leq e_G(v)$, for all $v \in V$, 
and two spanning trees $T$, one constructible in $O(\delta|E|)$ time and the other in $O(|E|)$ time,
which satisfy $e_G(v) \leq e_T(v) \leq e_G(v) + 4\delta + 1$ and
$e_G(v) \leq e_T(v) \leq 6\delta$, respectively.
Thus, the eccentricity terrain of a tree gives a good approximation (up-to an additive error $O(\delta))$ of the eccentricity terrain of a $\delta$-hyperbolic graph.
Furthermore, we obtain an approximation for the distance from an arbitrary vertex~$v$ to $C(G)$ or $C_{\leq 2\delta}(G)$ based on the eccentricity of~$v$.


\subsection{Related Works}

The eccentricity function/terrain has been studied extensively in Helly graphs, chordal graphs, $(\alpha_1, \triangle)$-metric graphs, and distance-hereditary graphs~\cite{FDraganPhD,Dragan2018,ourManuscriptDHG,Dragan2017EccentricityAT,DBLP:journals/dm/Prisner00a,Chepoi2018FastAO}, among others.
In~\cite{FDraganPhD}, it is shown that the eccentricity function in Helly graphs exhibits unimodality:
every vertex $v \notin C(G)$ has $loc(v)=1$.
In other words, any non-central vertex~$v$ has a shortest path~$P$ to a closest central vertex wherein 
any vertex on $P$ appears in a strictly lower eccentricity layer than the previous vertex until  $C(G)$ is reached.
Thus, any local minimum of the eccentricity function $e_G(v)$ coincides with the global minimum on Helly graphs~\cite{FDraganPhD}.
It is shown~\cite{FDraganPhD} that in such cases for any vertex $v \in V$, $e_G(v) = d(v,C(G)) + rad(G)$ holds.
Additionally, $(\alpha_1, \triangle)$-metric graphs, which include chordal graphs and the underlying graphs of 7-systolic complexes, have a similar but slightly weaker property.
In~\cite{Dragan2017EccentricityAT}, it is shown that every vertex $v \notin C(G)$ of a $(\alpha_1, \triangle)$-metric graph $G$ either has
$loc(v)=1$ or $e_G(v) = rad(G) + 1$, $diam(G) = 2rad(G)$, and $d(v,C(G)) = 2$.
So, any non-central vertex~$v$ has a shortest path~$P$ to a closest central vertex upon which the eccentricity of each vertex~$u \in P$ monotonically decreases until $C_1(G)$ and, furthermore, $|P \cap C_1(G)| \leq 2$.
The same behavior of the eccentricity function has recently been shown to exist in distance-hereditary graphs as well~\cite{ourManuscriptDHG}.
 This leads to a linear time additive 2-approximation for all eccentricities in chordal graphs via careful construction of a spanning tree~\cite{Dragan2017EccentricityAT,Dragan2018}
and a linear time additive 1-approximation for all eccentricities in a distance-hereditary graph via distances from a sufficient subset of central vertices~\cite{ourManuscriptDHG}.

As chordal graphs and distance-hereditary graphs are 1-hyperbolic, we question if the descending behavior of the eccentricity function persists in any $\delta$-hyperbolic graph.
Similar locality results have been established~\cite{Alrasheed_2016}:
any vertex~$v$ in a $\delta$-hyperbolic graph has either $loc(v) \leq 2\delta + 1$ or
$C(G)$ belongs to the set of vertices that are at most $4\delta+1$ from~$v$.
A pioneering work \cite{Chepoi_2008} first showed that, in a $\delta$-hyperbolic graph, $diam(G)$ and $2rad(G)$ are within $4\delta+1$ from each other and that the diameter of $C(G)$ in $G$ is at most $4\delta+1$. It gave also fast approximation algorithms for computing the diameter and the radius of $G$ and showed that there is a vertex $c$ in $G$, computable in linear time, such that each central vertex of $G$ is within distance at most $5\delta +1$ from $c$. Later in~\cite{Dragan2018RevisitingRD}, a better approximation algorithm for the radius was presented and a bound on the diameter of set $C_{\le 2\delta}(G)$ was obtained, namely,  $diam(C_{\le 2\delta}(G))\leq 8\delta+1$. Recently, similar results were obtained in \cite{Chepoi2018FastAO} for a related class of graphs, so called graphs with $\tau$-thin geodesic triangles (see Section \ref{section:auxLemmas} for a definition).
Additionally to approximating the diameter and the radius, \cite{Chepoi2018FastAO} gave efficient algorithms for approximating all eccentricities in such graphs via careful construction of a spanning tree.
We will mention the relevant results from~\cite{Dragan2018RevisitingRD,Chepoi_2008,Chepoi2018FastAO} in appropriate places later and compare them with our new results.

\hhjcss{
Note also that, under plausible assumptions,
even distinguishing the radius~\cite{10.5555/2884435.2884463}
or the diameter~\cite{Roditty_2013}
between exact values 2 or 3
cannot be accomplished in subquadratic time for sparse graphs.
Since the graphs constructed in the reductions~\cite{10.5555/2884435.2884463,Roditty_2013} are 1-hyperbolic, the same result holds for 1-hyperbolic graphs. 
Therefore, we are interested in fast approximation algorithms with additive errors depending linearly only on the hyperbolicity.
}

\section{Preliminaries}\label{section:auxLemmas}
All graphs occurring in this paper are connected, finite, unweighted, and undirected.
The \emph{length} of a path from a vertex~$u$ to a vertex~$v$ is the number of edges in the path.
The \emph{distance} $d_G(u,v)$ between two vertices~$u$ and~$v$ is the length of a shortest path connecting them in $G$.
We define the distance from a vertex~$v$ to a set $M \subseteq V$ of vertices as $d(v,M) = \min\{d(v,u) : u \in M\}$.
The \emph{eccentricity} $e_G(u)$ of a vertex~$u$ is the maximum distance from~$u$ to any other vertex in $G$, i.e., $e_G(u) = \max_{v \in V}d_G(u,v)$.
We omit the subindex when $G$ is known by context.
A graph's \emph{radius} $rad(G)$ is the minimum eccentricity of all vertices, and a graph's diameter $diam(G)$ is the maximum eccentricity.
The \emph{interval} between two vertices $x,y \in V$ is defined as the set of all vertices from any shortest $(x,y)$-path, that is, $I(x,y) = \{v \in V: d(x,v) + d(v,y) = d(x,y)\}.$
An interval \emph{slice} $S_k(x,y)$ is the set of vertices $\{v \in I(x,y) : d(v,x) = k\}$.
A \emph{disk} of radius~$k$ centered at a set~$S$ (or a vertex) is the set of vertices of distance at most~$k$ from~$S$, that is, $D(S,k) = \{u \in V: d(u,S) \leq k\}$.
We denote the set of furthest vertices from~$v$ as $F(v) = \{u \in V : d(u,v) = e(v)\}$. A pair $\{x,y\}$ of vertices is called 
\emph{a mutually distant pair} if $x\in F(y)$ and $y\in F(x)$. 
The \emph{diameter of a set} $S\subseteq V$ of a graph $G$ is $diam(S)=\max_{x,y\in S}d_G(x,y)$. 
The \emph{Gromov product} of two vertices~$x,y \in V$ with respect to a third vertex~$z \in V$ is defined as $(x|y)_z = \frac{1}{2}(d(x,z) + d(y,z) - d(x,y))$.
A list of these notations can be found in Appendix.

Let $S$ be a set and
let function $\hat{f} : S \rightarrow \mathbb{R}$ be an approximation of function $f : S \rightarrow \mathbb{R}$.
We say that $\hat{f}$ is an \emph{left-sided  additive $\epsilon$-approximation}  of $f$ if, for all $x \in S$, $f(x) - \epsilon \leq \hat{f}(x) \leq f(x)$.
We say that $\hat{f}$ is an \emph{right-sided additive $\epsilon$-approximation} of $f$ if, for all $x \in S$, $f(x) \leq  \hat{f}(x)\leq f(x) + \epsilon$. The value $\epsilon$ is called \emph{left-sided} (\emph{right-sided}, respectively) \emph{ additive error}. 
In a graph~$G$, a left-sided error appears when a vertex is returned by an algorithm whose eccentricity is an approximation of the diameter of~$G$ (as its eccentricity cannot exceed $diam(G)$),
whereas a right-sided error appears when a vertex is returned by an algorithm whose eccentricity is an approximation of the radius of~$G$ (as its eccentricity cannot be smaller than $rad(G)$).

\begin{figure}[!htb]
\begin{center}
\includegraphics[scale=0.64]{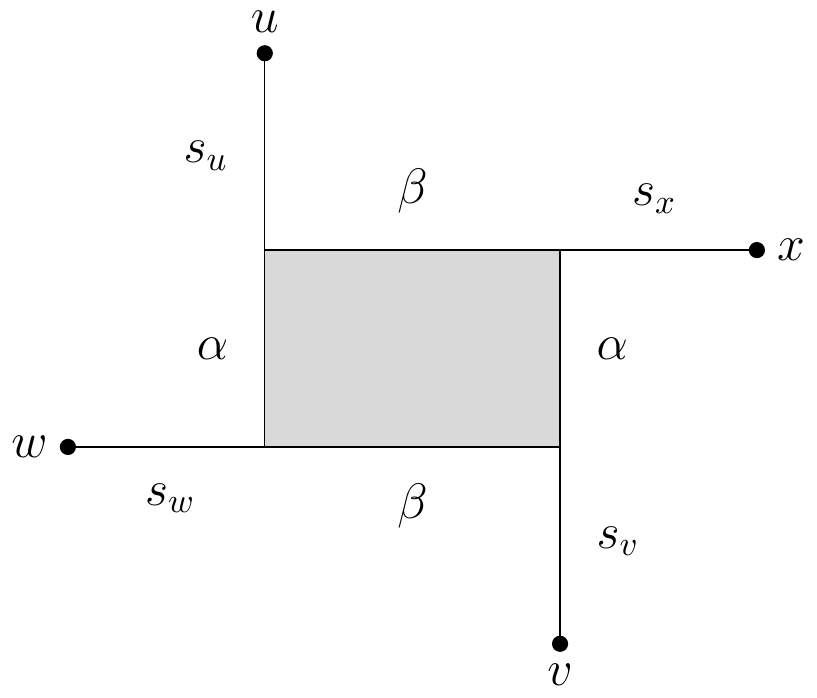}
\hspace*{1cm}
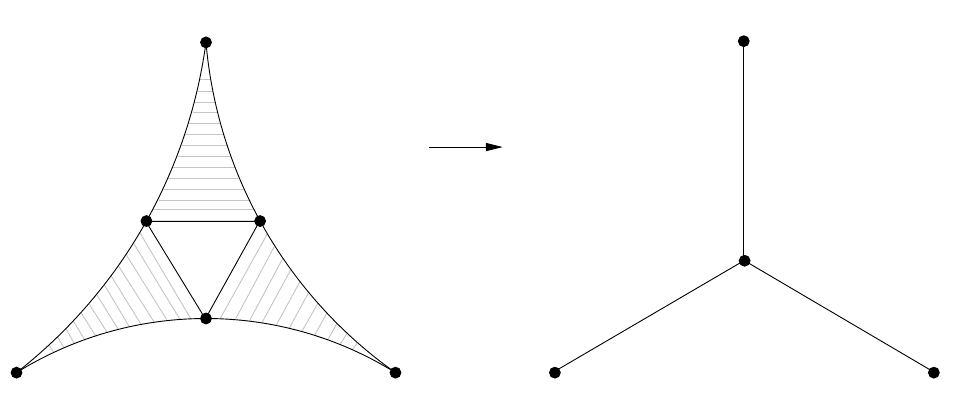
\end{center}
\caption{Illustration to the definitions of $\delta$: realization of the 4-point condition in the rectilinear plane (left), and a geodesic triangle $\Delta(x,y,z),$ the points $m_x, m_y,
m_z,$ and the tripod $T(x,y,z)$ (right).} \label{fig:thinTriangles}
\end{figure}

For metric spaces $(X,d)$, there are several equivalent definitions of $\delta$-hyperbolicity  with different but comparable values of~$\delta$~\cite{Alonso_1991aa,Bridson_1999,gromov90,Gromov1987}.
%
In this paper, we will use Gromov's 4-point condition: for any four points $u,v,w,x$ from $X$ the two larger of the three distance sums $d(u,v) + d(w,x)$, $d(u,x) + d(v,w)$, and $d(u,w) + d(v,x)$ differ by at most $2\delta \geq 0$.
A connected graph equipped with the standard graph metric $d_G$ is $\delta$-hyperbolic if the metric space $(V,d_G)$ is $\delta$-hyperbolic. The smallest value~$\delta$ for which~$G$ is $\delta$-hyperbolic is called the \emph{hyperbolicity} of~$G$ and is denoted by $\delta(G)$. Note that $\delta(G)$ is an integer or a half-integer. 
Every 4-point metric $d$ has a canonical representation in the rectilinear plane.
In Figure~\ref{fig:thinTriangles}, the three distance sums are ordered from large to small, implying that $\alpha \leq \beta$.
Then $\beta$ is half the difference of the largest minus the smallest sum, while $\alpha$ is half the difference of the largest minus the medium sum.
Hence, a metric space $(X,d)$ is $\delta$-hyperbolic if $\alpha \leq \delta$ for any four points $u,v,w,x \in X$.
At times we will compare our results to those known from literature, including those known for graphs defined by thin geodesic triangles as follows.

%
Let $(X,d)$ be a metric space.
An $(x,y)$-geodesic is a (continuous) map $\gamma$ from the segment $[a,b]$ of $\mathbb{R}^1$ to $X$ such
that $\gamma(a)=x$, $\gamma(b)=y$, and $d(\gamma(s), \gamma(t)) = |s - t|$ for all $s,t \in [a,b]$.
A metric space $(X,d)$ is geodesic if every pair of points in $X$ can be joined by a geodesic.
A geodesic triangle $\Delta(x,y,z)$ with $x,y,z \in X$ is defined on a geodesic metric space as
the union $[x,y] \cup [x,z] \cup [y,z]$ of three geodesic segments connecting $x,y,z$.
Let $m_x$ be the point of the geodesic segment $[y,z]$ located at distance $\alpha_y = (x|z)_y$ from $y$.
Then, $m_x$ is located at distance $\alpha_z = (x|y)_z$ from $z$ because $\alpha_y + \alpha_z = d(y,z)$.
Analogously, define the points $m_y \in [x,z]$ and $m_z \in [x,y]$ both located at distance $\alpha_x = (y|z)_x$ from~$x$;
see Figure~\ref{fig:thinTriangles} for an illustration.
There is a unique isometry $\varphi$ which maps $\Delta(x,y,z)$ to a tripod $T(x,y,z)$
consisting of three solid segments $[x,m]$, $[y,m]$, and $[z,m]$ of lengths $\alpha_x$, $\alpha_y$, and $\alpha_z$, respectively.
This function maps the vertices $x,y,z$ of $\Delta(x,y,z)$ to the respective leaves of $T(x,y,z)$
and the points $m_x$, $m_y$, and $m_z$ to the center $m$ of $T(x,y,z)$.
Any other point of $T(x,y,z)$ is the image of exactly two points of $\Delta(x,y,z)$.
A geodesic triangle is called $\delta$-thin if for all points $u,v \in \Delta(x,y,z)$, $\varphi(u) = \varphi(v)$ implies $d(u,v) \leq \delta$.
A graph~$G$ is $\delta$-thin if all geodesic triangles in~$G$ are $\delta$-thin.
The smallest value~$\delta$ for which~$G$ is $\delta$-thin is called the \emph{thinness} of~$G$ and is denoted by $\tau(G)$.

The thinness and hyperbolicity of a graph are comparable as follows (similar inequalities are known for general geodesic metric spaces).

\begin{proposition}\label{prop:thinVsHyperbolicity}~\cite{Alonso_1991aa,Bridson_1999,gromov90,Gromov1987}
For a graph $G$, $\delta(G) 
\leq \tau(G) \leq 4\delta(G)$, and the inequalities are sharp. 
\end{proposition}

We will often use the following lemma.
\begin{lemma}\label{lemDuality}
Let $G$ be a $\delta$-hyperbolic graph. For any $x,y,v \in V$ and any vertex $c \in I(x,y)$ the following holds. 
\setlist{nolistsep}
\begin{enumerate}[noitemsep, label=\roman*)]
\item[$(i)$] If $d(x,c) \leq (v|y)_x$, then $d(c,v) \leq d(x,v) - d(x,c) + 2\delta$ and $d(c,v) \leq d(x,v) + \delta$. Moreover, $e(c) \leq e(x) - d(x,c) + 2\delta$ and $e(c) \leq e(x) + \delta$, when $v\in F(c)$.
\item[$(ii)$] If $d(y,c) \leq (v|x)_y$, i.e., $d(x,c) \geq (v|y)_x$, then $d(c,v) \leq d(y,v) - d(y,c) + 2\delta$ and $d(c,v) \leq d(y,v) + \delta$. Moreover, $e(c) \leq e(y) - d(y,c) + 2\delta$ and $e(c) \leq e(y) + \delta$, when $v\in F(c)$.
\end{enumerate}
\end{lemma}
\begin{proof}
Consider the three distance sums $d(x,y) + d(c,v)$, $d(x,v) + d(c,y)$, and $d(x,c) + d(v,y)$.
Since $c \in I(x,y)$, $d(x,y) = d(x,c) + d(c,y)$.
By the triangle inequality, $d(x,v) + d(c,y) \leq d(x,c) + d(c,v) + d(c,y)$ and $d(x,c) + d(v,y) \leq d(x,c) + d(v,c) + d(c,y)$.
Thus, $d(x,y) + d(c,v) \geq \max\{d(x,v) + d(c,y),\ d(x,c) + d(v,y)\}$.

Suppose $d(x,c) \leq (v|y)_x = \frac{1}{2}(d(v,x) + d(y,x) - d(v,y))$.
We have $2d(x,c) + 2d(v,y) \leq (d(v,x) + d(y,x) - d(v,y)) + 2d(v,y) = d(v,x) + d(y,x) + d(v,y) = d(v,x) + d(x,c) + d(c,y) + d(v,y)$.
Subtracting $d(x,c) + d(v,y)$ from this inequality, one obtains $d(x,c) + d(v,y) \leq d(v,x) + d(c,y)$.
Since $G$ is $\delta$-hyperbolic,
$2\delta \geq (d(x,y) + d(c,v)) - (d(v,x) + d(c,y)) = d(x,c) + d(c,v) - d(v,x)$.
Therefore, $d(c,v) \leq d(v,x) - d(x,c) + 2\delta$.
By adding the triangle inequality $d(c,v) \leq d(v,x) + d(x,c)$ to this,
we obtain $d(c,v) \leq d(v,x) + \delta$.
Applying both inequalities to the case in which $v$ is furthest from~$c$, we get $e(c) = d(c,v) \leq d(v,x) - d(x,c) + 2\delta \leq e(x) - d(x,c) + 2\delta$
and $e(c) = d(c,v) \leq d(v,x) + \delta \leq e(x) + \delta$. Thus, $(i)$ is true. 

Suppose now that $d(x,c) \geq (v|y)_x = \frac{1}{2}(d(v,x) + d(y,x) - d(v,y))$.
First, we claim that this is equivalent to $d(c,y) \leq (v|x)_y$.
By assumption, $d(x,y) = d(x,c) + d(c,y) \geq \frac{1}{2}(d(v,x) + d(y,x) - d(v,y)) + d(c,y)$.
Therefore, $d(c,y) \leq d(x,y) - \frac{1}{2}(d(v,x) + d(y,x) - d(v,y)) = \frac{1}{2}(d(x,y) - d(v,x) + d(v,y)) = (v|x)_y$, establishing the claim.
Now, $(ii)$ is true by symmetry with $(i)$. 
\end{proof}

Lemma~\ref{lemDuality} has a few important corollaries. 

\begin{corollary}\label{cor:maxMinDuality}
Let~$G$ be a $\delta$-hyperbolic graph. Any $x,y,v \in V$ and $c \in I(x,y)$ satisfies $d(c,v) \leq \max\{d(x,v), d(y,v)\} - \min\{d(x,c), d(y,c)\} + 2\delta$.
\end{corollary}

We next combine both cases of Lemma~\ref{lemDuality} to form an upper bound on all distances
from vertex~$c$ on a shortest $(x,y)$-path, including $e(c)$,
as well as improvements to this bound when~$c$ is sufficiently far from the endpoints~$x$ and~$y$.
By these we generalize greatly some known results  from~\cite{Alrasheed_2016}.

\begin{corollary}\label{cor:distanceFromMiddle}
Let $G$ be a $\delta$-hyperbolic graph.
Any vertices $x,y,v \in V$ and $c \in I(x,y)$ satisfy $d(c,v) \leq \max\{d(x,v),\ d(y,v)\} + \delta$.
Furthermore, if $d(x,y) \geq 4\delta$, then any vertex $c^* \in I(x,y)$ with $d(x,c^*) \geq 2\delta$ and $d(y,c^*) \geq 2\delta$
satisfies $d(c^*,v) \leq \max\{d(x,v),\ d(y,v)\}$.
If $d(x,y) > 4\delta + 1$ then any vertex $c^* \in I(x,y)$ with $d(x,c^*) > 2\delta$ and $d(y,c^*) > 2\delta$
satisfies $d(c^*,v) < \max\{d(x,v),\ d(y,v)\}$~\cite{Alrasheed_2016}.
\end{corollary}
\begin{proof}
By Lemma~\ref{lemDuality}, $d(c,v) \leq d(x,v) + \delta$ or $d(c,v) \leq d(y,v) + \delta$.
Therefore, $d(c,v) \leq \max\{d(x,v),\ d(y,v)\} + \delta$.
If $d(x,y) \geq 4\delta$ then, by Corollary~\ref{cor:maxMinDuality}, any vertex $c^* \in I(x,y)$ with $d(x,c^*) \geq 2\delta$ and $d(c^*,y) \geq 2\delta$
satisfies $e(c^*) \leq \max\{d(x,v),d(y,v)\} - \min\{d(x,c^*),d(y,c^*)\} + 2\delta \leq \max\{d(x,v),d(y,v)\}$.
If $d(x,y) > 4\delta + 1$, i.e., $d(x,y) \geq 4\delta + 2$ then, by Corollary~\ref{cor:maxMinDuality}, any vertex $c^* \in I(x,y)$ with $d(x,c^*) > 2\delta$ and $d(c^*,y) > 2\delta$
satisfies $e(c^*) \leq \max\{d(x,v),d(y,v)\} - \min\{d(x,c^*),d(y,c^*)\} + 2\delta < \max\{d(x,v),d(y,v)\}$.
%
\end{proof}

\begin{corollary}\label{cor:deltaBoundOnEccentricity}
Let $G$ be a $\delta$-hyperbolic graph.
Any vertices $x,y \in V$ and $c \in I(x,y)$ satisfy $e(c) \leq \max\{e(x),\ e(y)\} + \delta$.
Furthermore, if $d(x,y) \geq 4\delta$, then any vertex $c^* \in I(x,y)$ with $d(x,c^*) \geq 2\delta$ and $d(y,c^*) \geq 2\delta$ satisfies $e(c^*) \leq \max\{e(x),\ e(y)\}$.
If $d(x,y) > 4\delta + 1$ then any vertex $c^* \in I(x,y)$ with $d(x,c^*) > 2\delta$ and $d(y,c^*) > 2\delta$ satisfies $e(c^*) < \max\{e(x),\ e(y)\}$~\cite{Alrasheed_2016}.
\end{corollary}
\begin{proof}
By Lemma~\ref{lemDuality}, $e(c) \leq e(x) + \delta$ or $e(c) \leq e(y) + \delta$.
Therefore, $e(c) \leq \max\{e(x),\ e(y)\} + \delta$.
Suppose that $d(x,y) \geq 4\delta$ and consider any vertex $c^* \in I(x,y)$ satisfying $d(x,c^*) \geq 2\delta$ and $d(c^*,y) \geq 2\delta$.
By Lemma~\ref{lemDuality}, $e(c^*) \leq e(x) - d(x,c^*) + 2\delta \leq e(x)$
or $e(c^*) \leq e(y) - d(y,c^*) + 2\delta \leq e(y)$.
Hence, $e(c^*) \leq \max\{e(x),\ e(y)\}$.

Suppose now that $d(x,y) > 4\delta + 1$, i.e., $d(x,y) \geq 4\delta + 2$.
Consider any vertex $c^* \in I(x,y)$ satisfying $d(x,c^*) > 2\delta$ and $d(c^*,y) > 2\delta$.
By Lemma~\ref{lemDuality}, $e(c^*) \leq e(x) - d(x,c^*) + 2\delta < e(x)$
or $e(c^*) \leq e(y) - d(y,c^*) + 2\delta < e(y)$.
Hence, $e(c^*) < \max\{e(x),\ e(y)\}$.
\end{proof}

\begin{corollary}\label{largerEccentricityOfMiddleWhenEndVerticesClose}
Let $G$ be a $\delta$-hyperbolic graph where $x,y \in V$, $d(x,y) \geq 2\delta+1$, and $c \in S_{2\delta+1}(x,y)$.
If $e(c) \geq \max\{e(x),\ e(y)\}$ then $d(x,y) \leq 4\delta + 1$.
\end{corollary}
\begin{proof} 
By contradiction assume that $e(c) \geq \max\{e(x),\ e(y)\}$ and $d(x,y) > 4\delta + 1$, i.e., $d(x,y) \geq 4\delta + 2$.
By Corollary \ref{cor:deltaBoundOnEccentricity}, 
$e(c) < \max\{e(x),\ e(y)\}$ must hold, giving  a contradiction. 
\end{proof}

\section{Pseudoconvexity of the sets $C_{\leq k}(G)$ and their diameters 
}\label{section:convexity}
A subset $S$ of a geodesic metric space or a graph is \emph{convex} if
for all $x,y \in S$ the metric interval $I(x,y)$ is contained in $S$.
This notion was extended by Gromov~\cite{Gromov1987} as follows:
for $\epsilon \geq 0$, a subset $S$ of a geodesic metric space or a graph is called \emph{$\epsilon$-quasiconvex} 
if for all $x,y \in S$ the metric interval $I(x,y)$ is contained in the disk $D(S,\epsilon)$. 
$S$ is said to be \emph{quasiconvex} if there is a constant $\epsilon \geq 0$ such that $S$ is $\epsilon$-quasiconvex.
Quasiconvexity plays an important role in the study of hyperbolic and cubical groups,
and hyperbolic graphs contain an abundance of quasiconvex sets~\cite{Chepoi:2017:CCI:3039686.3039835}.
Unfortunately, $\epsilon$-quasiconvexity is not closed under intersection.
Consider a path $P=(v_0,\dots, v_{2k})$ of length $2k$.  
Let $S_1 = \{v_0, v_{2k}\} \cup \{v_i : i \text{ is odd}\}$ and $S_2 = \{v_0, v_{2k}\} \cup \{v_i : i \text{ is even}\}$.
Both $S_1$ and $S_2$ are $1$-quasiconvex, however, their intersection is only $k$-quasiconvex. 

In this section, we introduce $\beta$-pseudoconvexity which satisfies this important intersection axiom of convexity
and we illustrate the presence of pseudoconvex sets in hyperbolic graphs.
For $\beta\ge 0$, we define a set $S \subseteq V$ to be  $\beta$-\emph{pseudoconvex} if, for any vertices $x,y \in S$, any vertex $z \in I(x,y) \setminus S$ satisfies $\min\{d(z,x), d(z,y)\} \leq \beta$. 
Note that when $\beta = 0$ the definitions of convex sets and $\beta$-pseudoconvex sets coincide.
Moreover, $\beta$-pseudoconvexity implies 
$\beta$-quasiconvexity.
Consider a $\beta$-pseudoconvex set $S$ and its arbitrary two vertices $x$ and $y$. As any vertex $z\in I(x,y)\setminus S$ satisfies 
$\min\{d(z,x), d(z,y)\} \leq \beta$, necessarily, $z$ belongs to disk $D(S,\beta)$.  
Since the empty set and $V$ are $\beta$-pseudoconvex, the following lemma establishes that $\beta$-pseudoconvex sets form a convexity. 

\begin{lemma}\label{pseudoconvexIntersection}
If sets $S_1 \subseteq V$ and $S_2 \subseteq V$ are $\beta$-pseudoconvex,
then $S_1 \cap S_2$ is $\beta$-pseudoconvex.
\end{lemma}
\begin{proof}
Consider any two vertices $x,y \in S_1 \cap S_2$.
If there is a vertex $z \in I(x,y)$ which does not belong to $S_1 \cap S_2$,
then,  $z \notin S_1$ or $z \notin S_2$. Without loss of generality, assume $z \notin S_1$. Then, $\min\{d(z,x), d(z,y)\} \leq \beta$ because $S_1$ is $\beta$-pseudoconvex.
\end{proof}

It is easy to see that in 0-hyperbolic graphs (which are block graphs, i.e., graphs in which every 2-connected component is a complete graph) all disks are convex. We next show that all disks are  $(2\delta -1)$-pseudoconvex in $\delta$-hyperbolic graphs with $\delta>0$.

\begin{lemma}\label{disksArePseudoconvex}
Let $G$ be a $\delta$-hyperbolic graph. 
Any disk of $G$ is $(2\delta-1)$-pseudoconvex, when $\delta>0$, and is convex, when $0\le \delta\le 1/2$.  
\end{lemma}
\begin{proof}
Consider a disk $D(v,r)$ centered at a vertex $v \in V$ and with radius $r$.
Let $x,y \in D(v,r)$ and let $z \in I(x,y)$ be a vertex which is not contained in $D(v,r)$.
By contradiction, assume that $d(z,x) \geq 2\delta$ and $d(z,y) \geq 2\delta$.
Since $z \notin D(v,r)$,  $d(v,z) > \max \{d(v,y), d(v,x)\} $.
By Corollary~\ref{cor:distanceFromMiddle} applied to vertices $x,y,v$ and vertex $z \in I(x,y)$,
necessarily,  $d(z,v) \leq \max\{ d(x,v), d(y,v) \}$, a contradiction.
\end{proof}

It is known that in chordal graphs (including 0-hyperbolic graphs) all  sets $C_{\leq k}(G)$, $k\in \mathbb{N}$, are convex (see, e.g.,~\cite{Chepoi-center-triangulated,Dragan2017EccentricityAT}). We next show that all such sets are  $(2\delta -1)$-pseudoconvex in $\delta$-hyperbolic graphs with $\delta>0$.

\begin{lemma}\label{centersArePseudoconvex}
Let $G$ be a $\delta$-hyperbolic graph  and $k \geq 0$ be an arbitrary integer. Any set $C_{\leq k}(G)$ of $G$ is $(2\delta -1)$-pseudoconvex, when $\delta>0$, and is convex, when $0\le \delta\le 1/2$.  
\end{lemma}
\begin{proof}
Let $S$ be the intersection of disks $D(v, rad(G) + k)$ centered at each vertex $v \in V$.
By Lemma~\ref{disksArePseudoconvex}, each disk is $(2\delta - 1)$-pseudoconvex.
By Lemma~\ref{pseudoconvexIntersection}, $S$ is also $(2\delta - 1)$-pseudoconvex.
It remains only to show that $S = C_{\leq k}(G)$.
Recall that $C_{\leq k}(G) = \{v \in V : e(v) \leq rad(G) + k\}$.
If $x \in S$, then $d(x,v) \leq rad(G) + k$ for all $v \in V$.
Therefore, $e(x) \leq rad(G) + k$ and so $x \in C_{\leq k}(G)$.
On the other hand, if $x \notin S$, then $d(x,v) > rad(G) + k$ for some $v \in V$.
Therefore, $e(x) > rad(G) + k$ and so $x \notin C_{\leq k}(G)$.
Hence, $S = C_{\leq k}(G)$. 
\end{proof}

As a consequence of Lemma~\ref{centersArePseudoconvex}, we obtain several interesting features of any shortest path between vertices of $C_{\leq k}(G)$
.
\begin{corollary}\label{pseudoconvexCentersA}
Let $G$ be a $\delta$-hyperbolic graph, and let $x,y \in C_{\leq k}(G)$ for an integer $k \geq 0$.
If there is a vertex $c \in I(x,y)$ where $c \notin C_{\leq k}(G)$, then $d(y,c) < 2\delta$ or $d(x,c) < 2\delta$.
\end{corollary}


\begin{corollary}\label{pseudoconvexCentersC}
Let $G$ be a $\delta$-hyperbolic graph with $\delta>0$, and let $x,y \in C_{\leq k}(G)$ for an integer $k \geq 0$.
If there is a shortest path $P(x,y)$ where $P(x,y) \cap C_{\leq k}(G) = \{x,y\}$, then $d(x,y) \leq 4\delta - 1$.
\end{corollary}
\begin{proof}
Assume $d(x,y) \geq 4\delta$ for some $x,y \in C_{\leq k}(G)$ and let $P(x,y)$ be a shortest path such that $P(x,y) \cap C_{\leq k}(G) = \{x,y\}$.
Consider vertex $c \in P(x,y)$ with $d(x,c) = 2\delta$.
Since $d(x,c) > 2\delta - 1$, by Lemma~\ref{centersArePseudoconvex}, necessarily $d(y,c) \leq 2\delta - 1$.
Thus, $d(x,y) = d(x,c) + d(c,y) \leq 2\delta + 2\delta -1 = 4\delta -1$, a contradiction.
\end{proof}

Note that for 0-hyperbolic graphs any shortest path $P(x,y)$ with $x,y\in C_{\leq k}(G)$ is contained in $C_{\leq k}(G)$ due to convexity of  $C_{\leq k}(G)$. 

We next obtain a bound on the diameter of set $C_{\leq k}(G)$.
It is known~\cite{Chepoi2018FastAO}  that if~$G$ is $\tau$-thin,
then $diam(C_{\leq k}(G)) \leq 2k + 2\tau + 1$.
Applying the inequality $\tau \leq 4\delta$ from Proposition~\ref{prop:thinVsHyperbolicity} yields $diam(C_{\leq k}(G)) \leq 2k + 8\delta + 1$,
which can be improved working directly with~$\delta$, hereby generalizing also a result from~\cite{Dragan2018RevisitingRD,Chepoi_2008}.

\begin{lemma}
Any $\delta$-hyperbolic graph~$G$ has $diam(C_{\leq k}(G)) \leq 2k + 4\delta + 1$ for every $k\in \mathbb{N}$.
In particular, $diam(C_{\leq 2\delta}(G)) \leq 8\delta + 1$~\cite{Dragan2018RevisitingRD}, 
$diam(C(G)) \leq 4\delta + 1$ 
and $diam(G) \geq 2rad(G) - 4\delta - 1$~\cite{Dragan2018RevisitingRD,Chepoi_2008}. 
\end{lemma}
\begin{proof}
Let $x,y \in C_{\leq k}(G)$ realize the diameter of $C_{\leq k}(G)$.
We have $e(x) \leq rad(G) + k$ and $e(y) \leq rad(G) + k$.
Consider a (middle) vertex $c \in I(x,y)$ so that $\min\{d(x,c), d(y,c)\} = \lfloor d(x,y) / 2 \rfloor$.
Let $v \in F(c)$ be a  vertex furthest from~$c$. Hence, $d(c,v) \geq rad(G)$.
By Corollary~\ref{cor:maxMinDuality},
$\lfloor d(x,y)/2 \rfloor = \min\{d(x,c),\ d(y,c)\} \leq \max\{d(v,x),\ d(v,y)\} - d(c,v) + 2\delta \leq rad(G) + k - rad(G) + 2\delta = k + 2\delta$.
Thus, $diam(C_{\leq k}(G))=d(x,y) = d(x,c) + d(c,y) \leq 2\lfloor d(x,y)/2 \rfloor + 1 \leq 2k + 4\delta + 1$. In particular, when $k=diam(G)-rad(G)$, we get $diam(G) \geq 2rad(G) - 4\delta - 1$ as $C_{\leq diam(G)-rad(G)}(G)=V$. 
\end{proof}
Thus, combining this with the result from~\cite{Chepoi2018FastAO}, we get  $diam(C_{\leq k}(G)) \leq 2k + 2\min\{\tau(G),2\delta(G)\} + 1$ for any graph $G$ and any $k\in \mathbb{N}$. 
\medskip 

Summarizing the results of this section,  we have.  
\begin{theorem}\label{th:convexity}
Every disk and every set $C_{\leq k}(G)$, $k \geq 0$, of a $\delta$-hyperbolic graph $G$ is $(2\delta -1)$-pseudoconvex, when $\delta>0$, and is convex, when $0\le \delta\le 1/2$. 
Furthermore, $diam(C_{\leq k}(G)) \leq 2k + 4\delta + 1$. 
\end{theorem}

For a $\delta$-hyperbolic graph $G$, although its center $C(G)$ has a bounded diameter in $G$, the graph $\langle C(G) \rangle$ induced by $C(G)$ may not be connected. This is the case even for distance-hereditary graphs (see, e.g., \cite{ourManuscriptDHG}) which are 1-hyperbolic. The following simple construction shows that even if the center of $G$ induces a connected subgraph, it may induce an arbitrary connected graph. Thus, even if $G$ has a bounded hyperbolicity, its center graph $\langle C(G) \rangle$ may have an arbitrarily large hyperbolicity.  
%
Consider any connected graph~$H$ with sufficiently large $\delta(H)$, and construct a new graph $G$ from $H$ by adding  four new vertices $x,y,x^*,y^*$ to $H$, making $x$ and $y$ adjacent to each vertex of~$G$, and making $x^*$ and $y^*$ adjacent only to $x$ and $y$, respectively.
It is easy to see that $G$ is 1-hyperbolic and $\langle C(G) \rangle$ is isomorphic to $H$. However, $H$ has a large hyperbolicity.

\section{Terrain shapes}\label{section:terrainShapes}
We consider the shape of a shortest path $P(y,x)$ as it travels from a vertex~$y$ to a vertex~$x$ through the eccentricity 
layers of $G$. 
We define an ordered pair of vertices $(u,v)$, where $(u,v) \in E$, as an \emph{up-edge} if $e(u) < e(v)$,
as a \emph{horizontal-edge} if $e(u) = e(v)$, and as a \emph{down-edge} if $e(u) > e(v)$.
Thus, any path $P(y,x) = (y = v_0, v_1, ..., v_k = x)$ from vertex $y$ to vertex $x$ can be described by a series of $k$ consecutive ordered pairs $(v_i,v_{i+1})$ which can be classified as either up-edges, down-edges, or horizontal-edges.
We define an \emph{m-segment} ($m$ stands for monotonic) as a series of consecutive ordered pairs 
$(v_{i}, v_{i+1},...,v_{i+\ell-1},v_{i+\ell})$ along a shortest path $P(y,x)$ in which each edge $(v_{j},v_{j+1})$ has the same classification.
An \emph{up-hill} (\emph{down-hill}) on $P(y,x)$ is a maximal by inclusion m-segment $(v_{i},\dots ,v_{i+\ell})$  of $P(y,x)$  where $(v_{j},v_{j+1})$ is an up-edge (down-edge) for each $j \in \{i,i+1,...,i+\ell-1\}$. The value $\ell$ is called the  \emph{height} of the hill. 
A \emph{plain} on $P(y,x)$ is a maximal by inclusion m-segment $(v_{i},\dots ,v_{i+\ell})$ of $P(y,x)$  where $(v_{j},v_{j+1})$ is a horizontal-edge for each $j \in \{i,i+1,...,i+\ell-1\}$. 
The value $\ell$ is called the \emph{width} of the plain.  
A plain $(v_{i},\dots ,v_{i+\ell})$ of $P(y,x)$ 
with $i>0$ and $i+\ell<k$  
is called a \emph{plateau} if  $e(v_{i-1})<e(v_{i})$ and $e(v_{i+\ell+1})<e(v_{i+\ell})$, is called a \emph{valley} if $e(v_{i-1})>e(v_{i})$ and  $e(v_{i+\ell+1})>e(v_{i+\ell})$, is called a \emph{terrace} if $e(v_{i-1})<e(v_{i})$ and $e(v_{i+\ell+1})>e(v_{i+\ell})$ or $e(v_{i-1})>e(v_{i})$ and $e(v_{i+\ell+1})<e(v_{i+\ell})$ 
(see Figure~\ref{fig:shapesExample} and Figure  \ref{fig:shapesExample2}).

\begin{figure}[!htb]
  \begin{center}
    \includegraphics[scale=0.8]{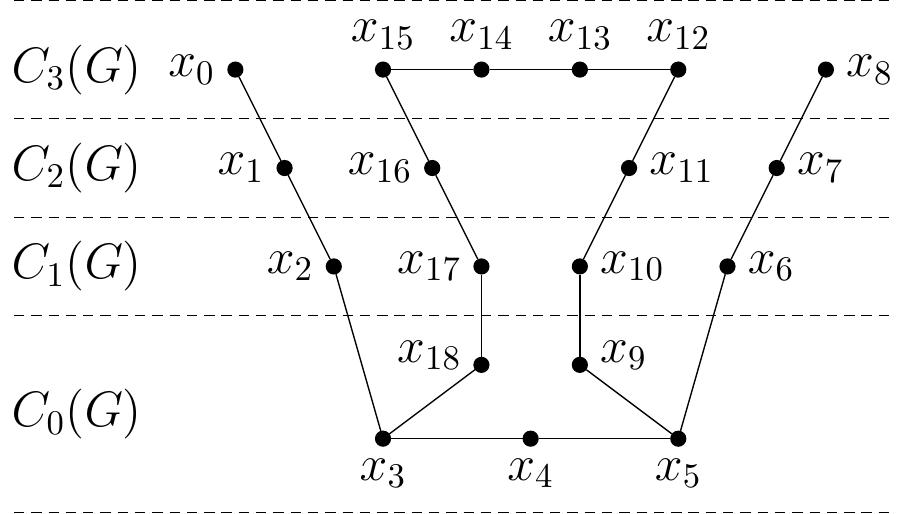}
    \caption{A $5/2$-hyperbolic graph $G$ with $rad(G)=6$ is shown with its eccentricity layers. Shortest path $P(x_0, x_{13})$ from $x_0$ to $x_{13}$ consists of down-hill $(x_0,x_1,x_2,x_3)$, valley $(x_3,x_{18})$, up-hill $(x_{18},x_{17},x_{16},x_{15})$, and 
    plain $(x_{15},x_{14},x_{13})$.  Shortest path $P(x_{16}, x_{11})$ from $x_{16}$ to $x_{11}$ consists of up-hill $(x_{16},x_{15})$, plateau $(x_{15},x_{12})$, down-hill $(x_{12},x_{11})$.
    }
    \label{fig:shapesExample}
  \end{center}
\end{figure}

In this section, we find a limit on the number of up-edges and horizontal-edges which can occur on a shortest path $P(y,x)$ from a vertex $y$ to a vertex $x$. Moreover, we discover that the length of any up-hill or the width of any plain of $P(y,x)$ is small and depends only on the hyperbolicity of $G$. As a consequence, we get that on any given shortest path $P$ from an arbitrary vertex to a closest central vertex, the number of vertices with locality more than 1  does not exceed $\max\{0,4\delta-1\}$. Furthermore, only at most $2\delta$ of them are  located outside $C_{\le \delta}(G)$ and only at most $2\delta+1$ of them are at distance $>2\delta$ from $C(G)$.  

First, for a shortest path $P(y,x)$ from $y$ to $x$ in an arbitrary graph, we establish a relation between the number of up-edges, down-edges, and the eccentricities of~$x$ and~$y$.
Let 
$\mathcal{U}(P(y,x))$, $\mathcal{H}(P(y,x))$, and $\mathcal{D}(P(y,x))$  denote respectively the number of up-edges, horizontal-edges, and down-edges along shortest path $P(y,x)$ when walking from $y$ to $x$.
Since each edge is classified as exactly one of the three categories,
any shortest path $P(y,x)$ has $d(y,x) = \mathcal{U}(P(y,x)) + \mathcal{H}(P(y,x)) + \mathcal{D}(P(y,x))$.

\begin{lemma}\label{numDownEdges}
Let $G$ be an arbitrary graph. 
For any shortest path $P(y,x)$ of $G$ from a vertex $y$ to a vertex $x$ the following holds:
$\mathcal{D}(P(y,x)) - \mathcal{U}(P(y,x)) =  e(y) - e(x)$,
that is, $e(x) + \mathcal{D}(P(y,x)) = \mathcal{U}(P(y,x)) + e(y)$.
\end{lemma}
\begin{proof}
We use an induction on $d(y,v)$ for any vertex $v \in P(y,x)$.
First assume that $v$ is adjacent to $y$.
If~$(y,v)$ is an up-edge, then $e(y) - e(v) = -1$ and $\mathcal{D}(P(y,v)) - \mathcal{U}(P(y,v)) = -1$.
If~$(y,v)$ is a horizontal-edge, then $e(y) - e(v) = 0$ and $\mathcal{D}(P(y,v)) - \mathcal{U}(P(y,v)) = 0$.
If~$(y,v)$ is a down-edge, then $e(y) - e(v) = 1$ and  $\mathcal{D}(P(y,v)) - \mathcal{U}(P(y,v)) = 1$.
Now consider an arbitrary vertex $v \in P(y,x)$ and assume, by induction, that $e(y) - e(v) = \mathcal{D}(P(y,v)) - \mathcal{U}(P(y,v))$.
Let vertex $u \in P(y,x)$ be adjacent to $v$ with $d(y,u) = d(y,v) + 1$.
By definition, $\mathcal{D}(P(y,u)) = \mathcal{D}(P(y,v)) + \mathcal{D}((v,u))$ and
$\mathcal{U}(P(y,u)) = \mathcal{U}(P(y,v)) + \mathcal{U}((v,u))$.
We consider three cases based on the classification of edge $(v,u)$.

If~$(v,u)$ is an up-edge, then $e(u) = e(v) + 1$, $\mathcal{U}((v,u)) = 1$, and $\mathcal{D}((v,u)) =0=  \mathcal{U}((v,u)) - 1$.
By the inductive hypothesis,
$\mathcal{D}(P(y,u)) = \mathcal{D}(P(y,v)) + \mathcal{D}((v,u))
                       = \mathcal{U}(P(y,v)) + e(y) - e(v) + \mathcal{U}((v,u)) - 1
                       = \mathcal{U}(P(y,u)) + e(y) - e(v) - 1
                       = \mathcal{U}(P(y,u)) + e(y) - e(u)$.

If~$(v,u)$ is a horizontal-edge, then $e(u) = e(v)$, $\mathcal{U}((v,u)) = 0=\mathcal{D}((v,u))$.
By the inductive hypothesis,
$\mathcal{D}(P(y,u)) = \mathcal{D}(P(y,v)) + \mathcal{D}((v,u))
                       = \mathcal{U}(P(y,v)) + e(y) - e(v) + \mathcal{U}((v,u))
                       = \mathcal{U}(P(y,u)) + e(y) - e(v)
                       = \mathcal{U}(P(y,u)) + e(y) - e(u)$.

If~$(v,u)$ is a down-edge, then $e(u) = e(v) - 1$, $\mathcal{U}((v,u)) = 0$, and  $\mathcal{D}((v,u)) = \mathcal{U}((v,u)) + 1$.
By the inductive hypothesis,
$\mathcal{D}(P(y,u)) = \mathcal{D}(P(y,v)) + \mathcal{D}((v,u))
                       = \mathcal{U}(P(y,v)) + e(y) - e(v) + \mathcal{U}((v,u)) + 1
                       = \mathcal{U}(P(y,u)) + e(y) - e(v) + 1
                       = \mathcal{U}(P(y,u)) + e(y) - e(u)$.
\end{proof}

\begin{lemma}\label{limitedAnomaliesToMonotonicity}
Let $G$ be an arbitrary graph.
For any shortest path $P(y,x)$ of $G$ from a vertex $y$ to a vertex $x$
the following holds: $2\mathcal{U}(P(y,x)) + \mathcal{H}(P(y,x)) = d(y,x) - (e(y) - e(x))$.
\end{lemma}
\begin{proof}
By Lemma~\ref{numDownEdges}, we have $\mathcal{D}(P(y,x)) = \mathcal{U}(P(y,x)) + e(y) - e(x)$.
Therefore,
$\mathcal{U}(P(y,x)) + \mathcal{H}(P(y,x)) + (\mathcal{U}(P(y,x)) + e(y) - e(x)) = \mathcal{U}(P(y,x)) + \mathcal{H}(P(y,x)) + \mathcal{D}(P(y,x)) = d(y,x)$.
Thus,  $2\mathcal{U}(P(y,x)) + \mathcal{H}(P(y,x)) = d(y,x) - (e(y) - e(x))$.
\end{proof}

In what follows, we focus on $\delta$-hyperbolic graphs.
First, we consider any  shortest path between two arbitrary vertices and show that any plain on it has width at most $4\delta+1$. If a plain is elevated and/or is far enough from the  end-vertices of the shortest path, then its width is even smaller.   

\begin{theorem}\label{th:plateauAny}
Let $G$ be a $\delta$-hyperbolic graph and
let $P(y,x)$ be a shortest path between any vertex $y\in V$ and  any vertex $x \in V$.
\begin{enumerate}[noitemsep, label=\roman*)]
\item[$(i)$]  Any plain $(u,...,v)$ of $P(y,x)$ has width $d(u,v) \leq 4\delta + 1$. 
Terraces are absent, if $\delta = 0$, and have width at most $4\delta - 1$, otherwise.
Plateaus are absent, if $\delta \leq \frac{1}{2}$, and have width at most $4\delta - 3$, otherwise.
\item[$(ii)$]  Any plain $(u,...,v)$ of $P(y,x)$ with $e(u) > \min\{e(x), e(y)\} + \delta$ has width $d(u,v) \leq 2\delta$,
and, if $(u,...,v)$ is a plateau, it has width at most $2\delta - 2$. \\
In particular,  if $\delta = 1$, then plateaus in any shortest path $P(y,x)$ are absent in eccentricity layers $C_k(G)$ for all $k > \min\{e(x), e(y)\} + 1 - rad(G)$.
Moreover, plateaus are completely absent if $\delta = 1$ and every vertex $c$ in $P(y,x)$, $c\neq x$, has $e(c)>e(x)$. 
\item[$(iii)$]  If there are two vertices $u,v \in P(y,x)$ with $e(u) = e(v)\ge \min\{e(x),e(y)\}$, $d(x, \{u,v\}) > 2\delta$ and $d(y, \{u,v\}) > 2\delta$, then $d(u,v) \leq 2\delta$.
\end{enumerate}
\end{theorem}
\begin{proof}
$(i)$ By definition of a plain, terrace, and plateau $(u,...,v)$, the eccentricity of each vertex $z \in (u,...,v)$ satisfies $e(z) = e(u)$.
In the case of terraces and plateaus, denote by $y'$ ($x'$) a vertex of $P(y,x)$ adjacent to $u$ (to $v$, respectively) but not in $(u,...,v)$.
Suppose, by contradiction, that $d(u,v) > 4\delta + 1$.
By Corollary~\ref{cor:deltaBoundOnEccentricity} applied to $u$ and $v$,
any vertex $c \in I(u,v)$ with $d(u,c) > 2\delta$ and $d(c,v) > 2\delta$
has eccentricity $e(c) < \max\{ e(u), e(v) \} = e(u)$. The latter contradicts with $e(z) = e(u)$. 

Let $(u,...,v)$ be a terrace of $P(y,x)$ and, without loss of generality, let $e(x') = e(u) - 1$.
Suppose, by contradiction, that $d(u,v) > 4\delta $, i.e., $d(u,x') > 4\delta + 1$.
By Corollary~\ref{cor:deltaBoundOnEccentricity} applied to $u$ and $x'$,
we again obtain a contradiction with $e(z) = e(u)$.
Hence, terraces have width at most $4\delta$. 
If additionally $\delta > 0$, by contradiction, suppose that $d(u,v) = 4\delta$.
Let $c \in (u,...,v)$ be a vertex at distance $2\delta - 1$ from $v$, i.e., $d(c,x') = 2\delta$ and $d(u,c) = 2\delta+1$.
We apply Lemma~\ref{lemDuality} to a $(u,x')$-subpath of $P(y,x)$ 
and to a vertex $w \in F(c)$.
If $d(x',c) \leq (w|u)_{x'}$, then $e(c) \leq e(x') - d(x',c) + 2\delta = e(c) - 1$, a contradiction.
If $d(x',c) \geq (w|u)_{x'}$, then $e(c) \leq e(u) - d(u,c) + 2\delta = e(c) - 1$, a contradiction.%

Let now $(u,...,v)$ be a plateau of $P(y,x)$. 
By maximality of m-segment $(u,...,v)$ and the definition of a plateau, $x',y'\in C_{e(u)-1}(G)$. As $(y',u,...,v,x') \cap C_{\le e(u)-1}(G) = \{x',y'\}$, by  Corollary~\ref{pseudoconvexCentersC}, $d(x',y') \leq 4\delta - 1$, implying $d(u,v) \leq 4\delta - 3$.
As $C_{\le e(u)-1}(G)$ is convex when $\delta \leq \frac{1}{2}$, plateaus can occur only if $\delta \geq 1$.

$(ii)$ Without loss of generality, assume that $e(x)\le e(y)$, i.e.,  $e(u) > e(x)+ \delta$.
By definition of a plain, the eccentricity of each vertex $z \in (u,...,v)$ satisfies $e(z) = e(u)$.
By contradiction, assume that $d(u,v) > 2\delta$.
Let $c \in (u,...,v)$ be a vertex at distance $2\delta +1$ from $u$.
We apply Lemma~\ref{lemDuality} to a $(u,x)$-subpath of $P(y,x)$ containing $(u,...,v)$ and to a vertex $w\in F(c)$.
If $d(x,c) \leq (w|u)_x$, then $e(u) = e(c) \leq e(x) + \delta$, a contradiction.
On the other hand, if $d(x,c) \geq (w|u)_x$, then $e(u) = e(c) \leq e(u) - d(u,c) + 2\delta = e(u) - 1$, a contradiction.

Let now $(u,...,v)$ be a plateau of $P(y,x)$. By $(i)$, $\delta \geq 1$.
By contradiction, assume that $d(u,v) > 2\delta - 2$.
Apply Lemma~\ref{lemDuality} to a $(y',x)$-subpath of $P(y,x)$ containing $(u,...,v)$ and to a vertex $w\in F(v)$.
If $d(x,v) \leq (w|y')_x$, then $e(u) = e(v) \leq e(x) + \delta$, a contradiction.
On the other hand, if $d(x,v) \geq (w|y')_x$, then $e(v) \leq e(y') - d(y',v) + 2\delta = e(v) - 1 - d(u,v) - 1 + 2\delta < e(v)$, a contradiction.
Consider now the case when $\delta = 1$.
By the previous claim, any plateau can occur only at eccentricity layer $C_k(G)$ for $k \leq \min\{e(x), e(y)\} + \delta -rad(G)= e(x) + 1-rad(G)$.
Finally, suppose that every vertex $c$ in $P(y,x)$, $c\neq x$, has $e(c)>e(x)$ and 
there is a plateau $(u,...,v)$ at eccentricity layer $C_k(G)$.
By maximality of m-segment $(u,...,v)$ and the definition of a plateau, $e(y') = k - 1 +rad(G) \leq e(x)$,
a contradiction with $e(x) < e(z)$ for all $z \in P(y,x)$.

$(iii)$ Assume $d(x, \{u,v\}) > 2\delta$ and $d(y, \{u,v\}) > 2\delta$ for vertices $u,v \in P(y,x)$ with $e(u) = e(v)$.
Without loss of generality, let $d(x,v) \leq d(x,u)$ and $e(x)\le e(y)$.
Consider an arbitrary $z \in F(v)$.
If $d(x,v) \leq (z|y)_x$ then, by Lemma~\ref{lemDuality}, $e(v) \leq e(x) - d(x,v) + 2\delta < e(x)$, and a contradiction with  $e(x)\le e(v)$ arises.
Thus, $d(x,v) > (z|y)_x$ and, by Lemma~\ref{lemDuality}, $d(y,u) \leq d(y,v) \leq (z|x)_y$ and  $e(v) \leq e(y) - d(y,v) + 2\delta = e(y) - (d(y,u) + d(u,v)) + 2\delta$.
By the triangle inequality, $e(v) = e(u) \geq e(y) - d(y,u)$.
Combining the previous two inequalities, we have $e(y) - d(y,u) \leq e(v) \leq e(y) - d(y,u) - d(u,v) + 2\delta$.
Therefore, $d(u,v) \leq 2\delta$.
\end{proof}


We define a shortest path $P(y,x)$ from a vertex~$y \in V$ to a vertex~$x \in V$
to be \emph{\minp{}} if $e(x)$ is minimal among all vertices of  $P(y,x)$, that is, all $v \in P(y,x)$ satisfy $e(x) \leq e(v)$.
$P(y,x)$ is referred to as \emph{\strictMinp} if all $v \in P(y,x)$ with $v \neq x$ satisfy $e(x) < e(v)$. 
Notice that any shortest path from an arbitrary vertex to a closest central vertex is strict end-minimal. 
We turn our focus now to \minp{} shortest paths because any shortest path can be decomposed
into two \minp{} subpaths.
Let $v \in P(y,x)$ be a vertex closest to $x$ of minimal eccentricity on shortest path $P(y,x)$.
Then $P(y,x)$ is represented by \minpath{} $P(y,v)$ joined with (strict) \minpath{} $P(x,v)$.

The following theorem shows that, in an  end-minimal shortest path $P(y,x)$ from $y$ to $x$,  all up-hills have bounded height, each vertex that is far from $x$ cannot have eccentricity higher than $e(y) + \delta$, and for each vertex $c$ that is far from the extremities of $P(y,x)$, all vertices $z$ of $P(y,x)$ which are between $y$ and  $c$ and with $d(c,z)\ge 2\delta+1$ have eccentricity larger than $e(c)$.   

\begin{theorem}\label{th:upHillsAnywhere}
Let $G$ be a $\delta$-hyperbolic graph
and let $P(y,x)=(y = v_0, v_1, ..., v_p = x)$ be an end-minimal shortest path from $y$ to $x$. 
\begin{enumerate}[noitemsep, label=\roman*)]
\item[$(i)$]  
Any up-hill $(u,...,v)$ of $P(y,x)$ has height $d(u,v) \leq \delta$.
\item[$(ii)$] Any $c \in P(y,x)$ with $d(c,x) > 2\delta$ satisfies $e(c) \leq e(y) + \delta$ and $e(c) \leq e(y) - d(c,y) + 2\delta$.
\item[$(iii)$] If $d(y,x) > 4\delta + 1$, then all vertices $v_i \in P(y,x)$ with $i \in [2\delta + 1, p - 2\delta - 1]$ have $e(v_i) < \min\{e(v_k) : k \in [0, i - 2\delta - 1]\}$ \emph{(a kind of pseudodescending)}.
\end{enumerate}
\end{theorem}
\begin{proof} $(i)$
Let $(u,...,v)$ be an m-segment on $P(y,x)$ which forms an up-hill.
By definition of an up-hill, eccentricity increases by one along each edge.
Therefore, $e(v) = e(u) + d(u,v)$.
By Corollary~\ref{cor:deltaBoundOnEccentricity} applied to a $(u,x)$-subpath of $P(y,x)$,
and because $e(x)$ is minimal on $P(y,x)$,
we have $e(u) + d(u,v) = e(v) \leq \max\{e(u),e(x)\} + \delta = e(u) + \delta$. 
Thus, $d(u,v) \leq \delta$.

$(ii)$ Let $v$ be an arbitrary vertex from $F(c)$.
Assume $d(x,c) \leq (v|y)_x$. By Lemma~\ref{lemDuality}, $e(c) \leq e(x) - d(x,c) + 2\delta < e(x)$,  a contradiction with $e(c)\ge e(x)$. 
Let now $d(x,c) > (v|y)_x$.
Then, by Lemma~\ref{lemDuality}, $e(c) \leq e(y) + \delta$ and $e(c) \leq e(y) - d(c,y) + 2\delta$.

$(iii)$ By contradiction assume that a vertex $v_i \in P(y,x)$ with $i \in [2\delta + 1, p - 2\delta - 1]$
has $e(v_i) \geq e(v_k)$ for some $k \in [0, i - 2\delta - 1]$.
Then, $d(v_k,v_i) \geq 2\delta + 1$ and $d(v_i,x) \geq 2\delta + 1$, and therefore $d(v_k,x) \geq 4\delta + 2$.
By Corollary~\ref{cor:deltaBoundOnEccentricity} applied to a subpath $P(v_k,x)$ of $P(y,x)$,
vertex $v_i \in P(v_k,x)$ with $d(v_k,v_i) > 2\delta$ and $d(v_i,x) > 2\delta$ satisfies $e(v_i) < \max\{e(v_k),\ e(x)\}$.
As $e(x)$ is minimal on $P(v_k,x)$, $e(v_i) < e(v_k)$, a contradiction.
\end{proof}

Theorem~\ref{th:upHillsAnywhere} in part $(iii)$ greatly generalizes a result from~\cite{Alrasheed_2016} where it was shown that any vertex~$v$ has $loc(v) \leq 2\delta + 1$ or $C(G) \subseteq D(v,4\delta+1)$. In particular, we have the following corollary. 

\begin{corollary}\label{cor:locality}
Let $G$ be a $\delta$-hyperbolic graph and 
$P(v,c)$ be a shortest path from an arbitrary vertex $v$ to an arbitrary central vertex $c\in C(G)$. Then, either the length of $P(v,c)$ is at most $4\delta+1$ or the vertex $u$ of $P(v,c)$ 
at distance $2\delta+1$ from $v$ satisfies $e(u)<e(v)$. 
\end{corollary}



An illustration of several results for \aMinpath{} $P(y,x)$ is shown in Figure~\ref{fig:shapesExample2}. \medskip

\begin{figure}[htb]
  \begin{center}
    \includegraphics[scale=0.7]{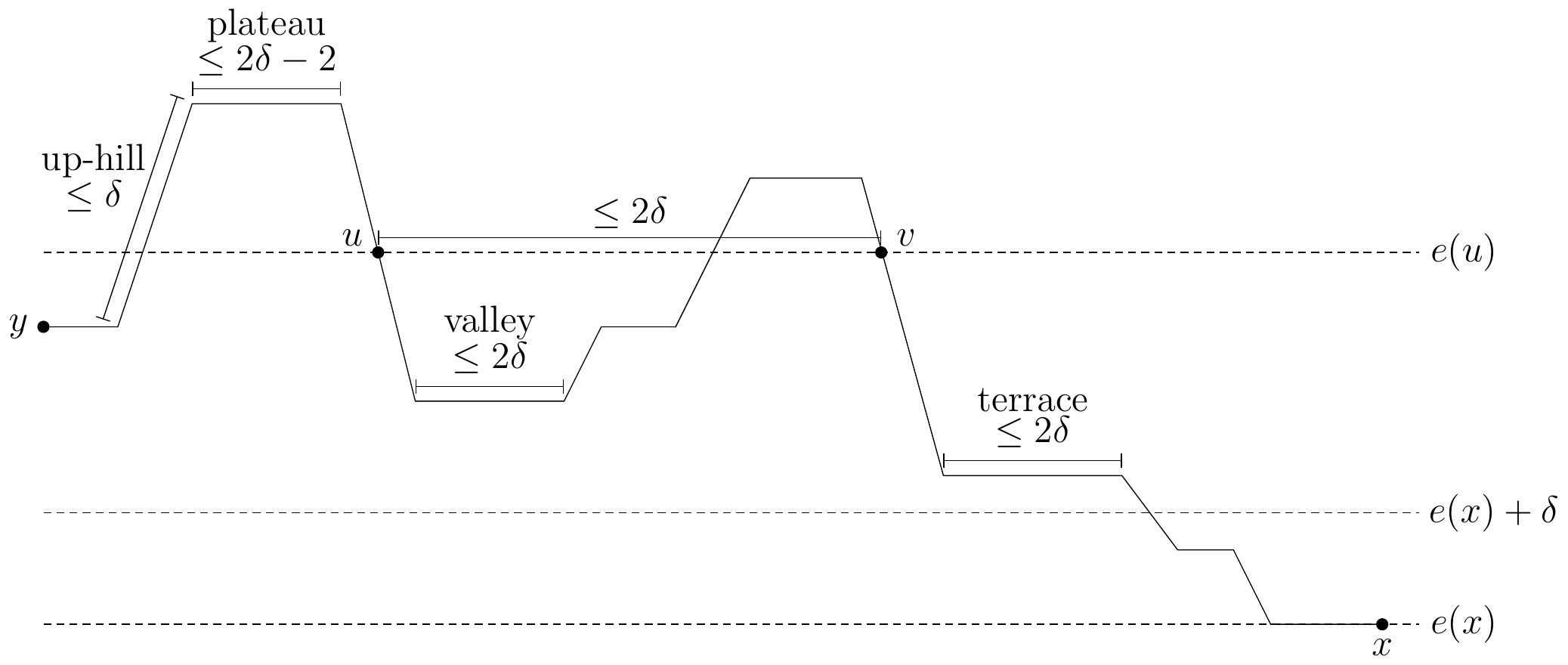}
    \caption{End-minimal shortest path $P(y,x)$ is depicted as it travels through the eccentricity layers of~$G$.
    By Theorem~\ref{th:plateauAny} and Theorem~\ref{th:upHillsAnywhere},  
    any up-hill on $P(y,x)$ has height at most $\delta$
    and any plain (including plateau, valley, or terrace) that is above the layer $C_{e(x)+\delta}(G)$ has width at most $2\delta$.
    Moreover, by Theorem~\ref{th:plateauAny}, two vertices~$u,v$ with the same eccentricity
    have distance at most $2\delta$ provided they are far (at least $2\delta+1$) from both end-vertices $y$ and $x$.
    }
    \label{fig:shapesExample2}
  \end{center}
\end{figure}


We next show that any \minpath{} $P(y,x)$ from $y$ to $x$ has no more than $4\delta+1$ up-edges and horizontal-edges combined.
Moreover, our result implies that, on any strict end-minimal shortest path $P(y,x)$
from an arbitrary vertex~$y$ to a vertex~$x$, the number of vertices with locality more than 1  does not exceed $4\delta$.
We give also two simple conditions which limit this number to $2\delta$.

\begin{lemma}\label{distanceToNearestSmallerEccentricity}

Let $G$ be a $\delta$-hyperbolic graph and $P(y,x)$ be a shortest path from $y$ to $x$. Then, the following holds:
\[
d(x,y) \leq e(y) - e(x) +
\begin{cases}
\max\{0,4\delta - 1\}, &\text{if  $P(y,x)$ is \strictMinp{},} \\
4\delta + 1, &\text{if $P(y,x)$ is \minp{}.}
\end{cases}
\]
Moreover, for a vertex $x'$ of $P(y,x)$, $d(y,x') \leq e(y) - e(x') + 2\delta$ if
$e(x') > e(x) + \delta$ or
$P(y,x)$ is \minp{} and $d(x,x') > 2\delta$.

\end{lemma}
\begin{proof}
Let $e(x) \leq e(c)$ for all $c \in P(y,x)$.
By contradiction, assume $d(x,y) \geq e(y) - e(x) + 4\delta + 2$.
As $e(y) - e(x) \geq 0$, $d(x,y) \geq 4\delta + 2$.
Pick a vertex~$c \in P(x,y)$ at distance $2\delta+1$ from $x$.
Let $v$ be an arbitrary vertex from $F(c)$.
If $d(x,c) \leq (v|y)_x$ then, by Lemma~\ref{lemDuality}, $e(c) \leq e(x) - d(x,c) + 2\delta = e(x) - 1$, a contradiction with $e(c) \geq e(x)$.
Hence, $d(x,c) > (v|y)_x$ and, by Lemma~\ref{lemDuality}, $e(c) \leq e(y) - d(y,c) + 2\delta$.
Therefore, $d(y,c) \leq e(y) - e(c) + 2\delta \leq e(y) - e(x) + 2\delta$.
Since $c \in I(x,y)$, we obtain $d(x,y) = d(x,c) + d(c,y) \leq 2\delta + 1 + e(y) - e(x) + 2\delta = e(y) - e(x) + 4\delta + 1$, a contradiction.

In the remaining case, when $P(y,x)$ is \strictMinp{},  
we apply similar arguments as above but with vertex~$c$ at distance $\ell$ from~$x$, where $\ell=1$ when $\delta =0$ and $\ell=2\delta$ when $\delta \ge 1/2$. We have $e(y)>e(x)$ and $e(c)>e(x)$. 
By contradiction, assume $d(x,y) \geq e(y) - e(x) + 4\delta$.
That is, $d(x,y) \geq 4\delta + 1$.
If $d(x,c) \leq (v|y)_x$ then, by Lemma~\ref{lemDuality}, $e(c) \leq e(x) - d(x,c) + 2\delta \le e(x)$, a contradiction with $e(c) > e(x)$.
If $d(x,c) > (v|y)_x$ then, by Lemma~\ref{lemDuality}, $e(c) \leq e(y) - d(y,c) + 2\delta$.
Therefore, $d(y,c) \leq e(y) - e(c) + 2\delta \leq e(y) - e(x)-1 + 2\delta$. Hence, $d(x,y) = d(x,c) + d(c,y) \leq d(x,c) + e(y) - e(x) + 2\delta -1$. That is, $d(x,y)  \le e(y) - e(x) + 4\delta-1$ when $\delta>0$ (contradicting our assumption), and $d(x,y)  \le e(y) - e(x) = e(y) - e(x) +4\delta$ when $\delta=0$.

Let now $x'$ be a vertex of $P(y,x)$ with $e(x') > e(x) + \delta$ or
$d(x,x') > 2\delta$ and $P(y,x)$ is \minp{}.
By Lemma~\ref{lemDuality}, either $e(x') \leq e(x) + \delta$ and $e(x') \leq e(x) - d(x,x') + 2\delta$ holds or $e(x') \leq e(y) - d(y,x') + 2\delta$ holds.
If the former case is true, necessarily,  $P(y,x)$ is \minp{},  $d(x,x') > 2\delta$ and $e(x') \leq e(x) - d(x,x') + 2\delta < e(x)$, contradicting with $P(y,x)$ being \minp{}.
In the latter case, $d(y,x') \leq e(y) - e(x') + 2\delta$.
\end{proof}

\begin{theorem}\label{thm:upAndHorEdges}
If $G$ is $\delta$-hyperbolic, then for any shortest path $P(y,x)$
from $y$ to $x$  
the following holds:
\[
2\mathcal{U}(P(y,x)) + \mathcal{H}(P(y,x)) \leq
\begin{cases}
\max\{0,4\delta - 1\}, &\text{if  $P(y,x)$ is \strictMinp{},} \\
4\delta + 1, &\text{if $P(y,x)$ is \minp{}.}
\end{cases}
\]
Moreover, for a vertex $x'$ of $P(y,x)$, $2\mathcal{U}(P(y,x')) + \mathcal{H}(P(y,x')) \leq 2\delta$
if
$e(x') > e(x) + \delta$ or
$P(y,x)$ is \minp{} and $d(x,x') > 2\delta$.
\end{theorem}
\begin{proof}
The proof follows directly from Lemma~\ref{limitedAnomaliesToMonotonicity} and Lemma~\ref{distanceToNearestSmallerEccentricity}.
\end{proof}

\begin{corollary}\label{cor:numberoflocality}
Let $G$ be a $\delta$-hyperbolic graph. Then, on any 
strict end-minimal shortest path $P(y,x)$ from a vertex~$y$ to a vertex~$x$,  the number of vertices with locality more than 1  does not exceed $4\delta$. If additionally $x\in C(G)$, then the number of vertices with locality more than 1  does not exceed $\max\{0,4\delta-1\}$. 
\end{corollary}
\begin{proof} 
As we go from $y$ to $x$ along $P(y,x)$, every vertex $u$ of $P(y,x)$, except $x$, is the beginning of an edge $(u,v)$ on $P(y,x)$. If an ordered pair  $(u,v)$ forms a down-edge, then the vertex $u$ has locality 1 in $G$. 
Only when an ordered pair $(u,v)$ forms an up-edge or a horizontal-edge on $P(y,x)$, then the vertex $u$ may have locality more than 1 in $G$. 
If $\delta>0$, then any strict end-minimal shortest path $P(y,x)$ has no more than $4\delta-1$ up-edges and horizontal-edges combined. Hence, together with $x$, there are at most $4\delta$  vertices on $P(y,x)$ with locality more than 1. If $\delta=0$, then $G$ is a block (and hence, a Helly) graph and each of its non-central vertices has locality 1~\cite{FDraganPhD}. Recall that, by definition, the locality of a central vertex is 0.  
%
\end{proof}

\begin{corollary}\label{cor:numberoflocality+}
Let $G$ be a $\delta$-hyperbolic graph. Then, on any 
shortest path $P(y,x)$ between a vertex~$y$ and a vertex~$x$,  the number of vertices with locality more than 1  does not exceed $8\delta+1$. If $P(y,x)$ is end-minimal, then the number of vertices with locality more than 1  does not exceed $4\delta+2$. 
\end{corollary}
\begin{proof} Let $P(y,x)$ be an end-minimal shortest path from $y$ to $x$. Using Theorem~\ref{thm:upAndHorEdges} and same arguments as in the proof of Corollary~\ref{cor:numberoflocality}, we get that at most $4\delta+1$ vertices of $P(y,x)\setminus\{x\}$ have locality more than 1. Hence, together with $x$, there are at most $4\delta+2$  vertices on $P(y,x)$ with locality more than 1.

Let now $P(y,x)$ be an arbitrary shortest path between $y$ and  $x$. Let also $v$ be a vertex from $P(y,x)$ with minimal  eccentricity closest to $x$. Then, subpath $P(x,v)$ of $P(x,y)$ is strict end-minimal and subpath $P(y,v)$ of $P(y,x)$ is end-minimal. There are at most $4\delta$ vertices with locality more than 1 in $P(x,v)$ and there are at most $4\delta+1$ vertices with locality more than 1 in $P(y,v)\setminus\{v\}$.  Thus, $P(y,x)$ has at most $8\delta+1$ such vertices. 
\end{proof}

These corollaries can be refined in the following way.
\begin{corollary}
Let $P(y,x)$ be a shortest path from any vertex~$y$ to any vertex~$x$.
Then, a (prefix)  subpath $P(y,x')$ of $P(y,x)$ has at most $k$ vertices with locality more than 1, where
\[
k =
\begin{cases}
2\delta, &\text{if $e(x') > e(x) + \delta$,} \\
2\delta+1, &\text{if $P(y,x)$ is \minp{} and $d(x,x') > 2\delta$.}
\end{cases}
\]
\end{corollary}
\begin{proof}
As, by
Theorem~\ref{thm:upAndHorEdges},
$2\mathcal{U}(P(y,x')) + \mathcal{H}(P(y,x')) \leq 2\delta$,
we can use same arguments as in the proof of Corollary~\ref{cor:numberoflocality} to show that
in $P(y,x') \setminus \{x'\}$ there are at most $2\delta$ vertices with locality more than 1.
It remains only to show that the entire $P(y,x')$ has at most  $2\delta$ vertices with locality more than 1  
when $e(x') > e(x) + \delta$.
Without loss of generality, we can pick a vertex $x' \in P(y,x)$, with $e(x') > e(x) + \delta$, that is furthest from~$y$.
By the choice of~$x'$, the neighbor~$x''$ of~$x'$ on $P(y,x)$ that  is closer to~$x$ satisfies $e(x'') = e(x) + \delta = e(x') - 1$.
Hence, $loc(x')=1$.
\end{proof}

Thus, on any shortest path from an arbitrary vertex to a closest central vertex, there are at most $\max\{0,4\delta-1\}$ vertices with locality more than 1, and only at most $2\delta$ of them are  located outside $C_{\le \delta}(G)$ and only at most $2\delta+1$ of them are at distance $>2\delta$ from $C(G)$.  

\medskip

In certain graph classes up-hills and plains 
are restricted in their location along a shortest path $P(y,x)$ connecting a vertex $y$ to a closest central vertex $x \in C(G)$.
Helly graphs contain no such non-descending shapes, whereas chordal graphs and distance-hereditary graphs have no up-hills but plains of width at most 1 may   occur~\cite{FDraganPhD,Dragan2017EccentricityAT,ourManuscriptDHG}. Furthermore, for every vertex $y$, there is a shortest path $P(y,x)$  from $y$ to any closest central vertex $x$ such that it has at most one plain of width 1, and if a plain exists then it is located in layer $C_1(G)$. 
We observe that in hyperbolic graphs even up-hills can occur anywhere on any shortest path - it can be close or far from a central vertex or endpoints of the path.

\begin{figure}[!htb]
  \begin{center}
    \includegraphics[scale=0.7]{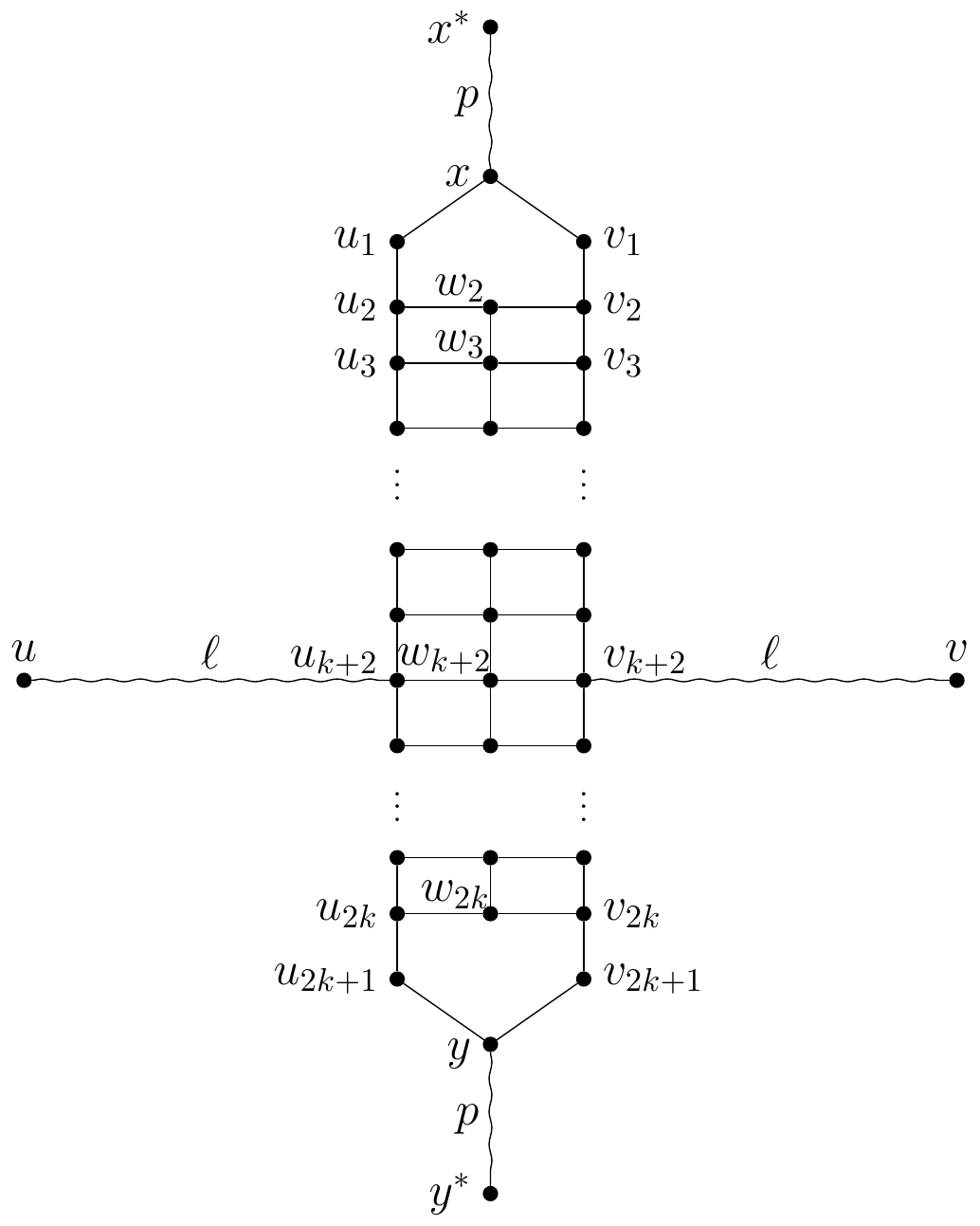}
    \caption{An illustration that in a 2-hyperbolic graph up-hills on each shortest path to the center or to a furthest vertex can occur very far from the center and from the endpoints of the path.}
    \label{fig:hillsAreAnywhere}
  \end{center}
\end{figure}

Consider a 2-hyperbolic graph $G=(V,E)$ depicted in Figure~\ref{fig:hillsAreAnywhere}. It has two paths
 $(x,u_1,u_2,...,u_{2k+1},y)$ and $(x,v_1,v_2,...,v_{2k+1},y)$ of length $2k+2$,
 a path $(w_2,w_3,...,w_{2k})$ of length $2k-2$,
 and edges $(u_i,w_i) \in E$ and $(w_i,v_i) \in E$ for each $i \in [2, 2k]$.
It has also two paths each of length $\ell$ connecting vertex $u_{k+2}$ to vertex $u$ as well as vertex $v_{k+2}$ to vertex $v$,
and two paths each of length $p > 0$ connecting $x$ to $x^*$ as well as $y$ to $y^*$.
If $\ell = k + p$, then $diam(G) = 2\ell + 2 = d(u,v) = d(x^*,y^*)$, $rad(G)=\ell+2$, and $C(G)=\{u_{k+2}, w_{k+1}, v_{k+2}\}$.
Observe that $e(x) = d(x,u) = \ell + k+2$ whereas $e(u_1) = d(u_1, v) = \ell+k+3$ and $e(v_1) = d(v_1,u) = \ell + k + 3$.
Any shortest $(x^*, z)$-path where $z \in C(G)$ or $z \in F(x^*)$ contains either the up-hill $(x,u_1)$ or the up-hill $(x,v_1)$.
However, both up-hills are arbitrarily far from the center $C(G)$ and far from any furthest vertex in $F(x^*)$.
Both up-hills also occur arbitrarily far from the starting vertex $x^*$ of the path.
Up-hills also occur on all shortest paths between the diametral pair $(x^*,y^*)$.

\section{Bounds on the eccentricity of a vertex}\label{section:eccentricityBounds}
In this section, we show that the auxiliary lemmas stated earlier yield several known from~\cite{Dragan2018RevisitingRD,Chepoi_2008} results on finding a vertex with small or large eccentricity, as well as intermediate results regarding
the relationship between diameter and radius.
We obtain also new efficient algorithms for approximating all vertex eccentricities in $\delta$-hyperbolic graphs and compare them with known results on graphs with $\tau$-thin triangles~\cite{Chepoi2018FastAO}. We present the following algorithms for approximating all eccentricities:
a $O(\delta|E|)$ time left-sided additive $2\delta$-approximation,
a $O(\delta|E|)$ time right-sided additive $(4\delta+1)$-approximation,
and a $O(|E|)$ time right-sided additive $6\delta$-approximation.

But first, we establish some lower and upper bounds on the  eccentricity of any vertex based on its distance to either $C(G)$ or $C_{\leq 2\delta}(G)$, and vice versa.

\subsection{Relationship between eccentricity of a vertex and its distance to $C(G)$ or $C_{\leq 2\delta}(G)$}
In this subsection, we show that the eccentricity of a vertex is closely related to its distance to both $C(G)$ and $C_{\leq 2\delta}(G)$,
analogous up to $O(\delta)$ to that of Helly graphs.
Recall that an interval slice $S_k(x,y)$ is the set of vertices $\{v \in I(x,y) : d(v,x) = k\}$.

We will need the following lemma which is a consequence of Lemma~\ref{lemDuality}. 
\begin{lemma}\label{cor:radiusBackFromBeam}
Let $G$ be a $\delta$-hyperbolic graph and let $x \in V$, $y \in F(x)$.
Any vertex $c \in S_{rad(G) + k}(y,x)$, $0\le k\le d(x,y)-rad(G)$, has $e(c) \leq rad(G) + 2\delta + k$.
In particular, $c \in S_{rad(G)}(y,x)$ has $e(c) \leq rad(G) + 2\delta$~\cite{Dragan2018RevisitingRD,Chepoi_2008}.
\end{lemma}
\begin{proof}
Let $v$ be any vertex from $F(c)$.
If $d(x,c) \leq (v|y)_x$ then, by Lemma~\ref{lemDuality}
and the fact that $d(v,x) \leq d(x,y)$, we have
 $e(c) = d(c,v) \leq d(x,v) - d(x,c) + 2\delta \leq d(x,y) - (d(x,y) - d(y,c)) + 2\delta = rad(G) + 2\delta + k$.
On the other hand, if $d(x,c) \geq (v|y)_x$ then, by Lemma~\ref{lemDuality}
and the fact that $2rad(G) \geq d(v,y)$, we have
 $e(c)= d(c,v) \leq d(y,v) - d(y,c) + 2\delta \leq 2rad(G) - rad(G) - k + 2\delta = rad(G) + 2\delta - k$.
\end{proof}

\begin{theorem}\label{thm:approxEccToC2Delta}
Let $G$ be a $\delta$-hyperbolic graph.
Any vertex $x$ of $G$ satisfies the following inequalities: $$d(x, C_{\leq 2\delta}(G)) + rad(G) + 2\delta \geq e(x) \geq d(x,C_{\leq 2\delta}(G)) + rad(G).$$
$$d(x,C(G)) + rad(G) - 4\delta\leq e(x) \leq d(x, C(G)) + rad(G).$$
\end{theorem}
\begin{proof}
Let $y \in F(x)$ be a furthest vertex from $x$, $c$ ($c'$) be a vertex closest to $x$ in $C(G)$ ($C_{\leq 2\delta}(G)$, respectively).
On one hand, by the triangle inequality, $e(x) = d(x,y) \leq d(x,c) + d(c,y)\le d(x,c) + e(c)=d(x,C(G)) + rad(G)$ and $e(x) = d(x,y) \leq d(x,c') + d(c',y) \leq d(x,c') + e(c')=d(v,C_{\leq 2\delta}(G)) + rad(G)+2\delta$. 
On the other hand, by Lemma~\ref{distanceToNearestSmallerEccentricity}, 
$d(x,c)\le e(x)-e(c)+4\delta$. 
Thus, $e(x)\ge d(x,c)+e(c)-4\delta=d(x,C(G)) + rad(G) - 4\delta$. Furthermore, by Lemma~\ref{cor:radiusBackFromBeam}, any vertex $c^*\in S_{rad(G)}(y,x)$ satisfies $e(c^*)\le rad(G)+2\delta$. Hence, $e(x)= d(x,c^*)+d(c^*,y)\ge d(x,C_{\leq 2\delta}(G)) + rad(G)$. 
\end{proof}

It is known~\cite{Chepoi2018FastAO} that if~$G$ is a $\tau$-thin graph, then for any vertex~$x \in V$, 
$d(x,C(G)) + rad(G) - 4\tau - 2 \leq e(x)$.
Applying the inequality $\tau \leq 4\delta$ from Proposition~\ref{prop:thinVsHyperbolicity} yields
$d(x,C(G)) + rad(G) - 16\delta - 2 \leq e(x)$.
Working directly with~$\delta$, in Theorem~\ref{thm:approxEccToC2Delta},  we obtained a significantly better bound with~$\delta$,
which, as $\delta \leq \tau$, also
improves the bound known with~$\tau$.

\begin{corollary}\label{cor:approxEccToCWithThinness}
Let $G$ be a $\tau$-thin graph.
Any vertex~$x$ satisfies the following inequality:
$$d(x,C(G)) + rad(G) - 4\tau \leq e(x) \leq d(x, C(G)) + rad(G).$$
\end{corollary}

Let $x$ be an arbitrary vertex with eccentricity $e(x) = rad(G) + k$ for some integer $k \geq 0$.
By Theorem~\ref{thm:approxEccToC2Delta}
, we have: 
\begin{align}
\begin{split}
  k \geq d(x,C_{\le 2\delta}(G)) \geq k- 2\delta  \label{eqn:a}
\end{split}
\\
\begin{split}
  k \leq d(x,C(G)) \leq k + 4\delta \label{eqn:b}
\end{split}
\end{align}

Hence, one obtains a relationship also between the distance from $x$ to $C(G)$ and to $C_{\leq 2\delta}(G)$.
\begin{corollary}
Let~$G$ be a $\delta$-hyperbolic graph and let~$x \in V$ with $e(x) = rad(G) + k$.
Then, $d(x,C_{\leq 2\delta}(G)) \leq k \leq d(x,C(G))$.
Moreover, $d(x,C_{\leq 2\delta}(G)) = \ell$ implies $d(x,C(G)) \leq \ell + 6\delta$.
\end{corollary}
\begin{proof}
Combining equations~(\ref{eqn:a}) and~(\ref{eqn:b}) yields $d(x,C_{\leq 2\delta}(G)) \leq k \leq d(x,C(G))$.
Assume now that $d(x,C_{\leq 2\delta}(G)) = \ell$.
By equation~(\ref{eqn:a}), $\ell \geq k - 2\delta$.
By equation~(\ref{eqn:b}), $d(x,C(G)) \leq k + 4\delta \le \ell + 6\delta$.
\end{proof}




Now, we turn our focus from end-minimal shortest  paths to  shortest $(x,y)$-paths wherein $y \in F(x)$ and $x\in F(z)$ for some vertex $z\in V$ or 
when $\{x,y\}$ is a mutually distant pair.

\subsection{Finding a vertex with small or large eccentricity 
and left-sided additive approximation of all vertex eccentricities}
Let $\{x,y\}$ be a pair of vertices such that $y \in F(x)$ and $x\in F(z)$ for some vertex $z\in V$. 
In Lemma~\ref{cor:radiusBackFromBeam}, we established that 
the eccentricity of vertex~$c_r$ on any shortest $(x,y)$-path at distance $rad(G)$ from~$y$ has small eccentricity (within $2\delta$ of radius).
Of more algorithmic convenience,
we show here that even a middle vertex~$c_m$ of any shortest $(x,y)$-path has small eccentricity (within $3\delta$ of radius),
and its eccentricity is even smaller (within $2\delta$ of radius) if $x \in F(y)$ as well, i.e., when $\{x,y\}$ is a mutually distant pair. 

We will need the following lemma from~\cite{Dragan2018RevisitingRD}. 
\begin{lemma}\cite{Dragan2018RevisitingRD}\label{generalDualityLemma}
Let $G$ be a $\delta$-hyperbolic graph. For every quadruple $c,v,x,y \in V$,
$d(x,v) - d(x,y) \geq d(c,v) - d(y,c) - 2\delta$ or $d(y,v) - d(x,y) \geq d(c,v) - d(x,c) - 2\delta$ holds.
\end{lemma}

Interestingly, the distances from any vertex $c$ to two mutually distant vertices give a very good estimation on the eccentricity of $c$. 

\begin{theorem} \label{thm:eccOnLineBetter}
Let $G$ be a $\delta$-hyperbolic graph, and let $\{x,y\}$ be a mutually distant pair of vertices.
Any vertex $c \in V$ has $\max\{d(x,c), d(y,c)\} \leq e(c) \leq \max\{d(x,c), d(y,c)\} + 2\delta$.
Moreover, any vertex $c^* \in S_{\lfloor d(x,y)/2 \rfloor}(x,y)\cup S_{\lfloor d(x,y)/2 \rfloor}(y,x)$ has $e(c^*) \leq \lceil d(x,y)/2 \rceil + 2\delta \leq rad(G) + 2\delta$~\cite{Dragan2018RevisitingRD}.
In particular, $diam(G) \geq d(x,y) \geq 2rad(G) - 4\delta - 1$~\cite{Chepoi_2008}.
\end{theorem}
\begin{proof}
The inequality $e(c) \geq \max\{d(x,c), d(y,c)\}$ holds for any three vertices by definition of eccentricity.
To prove the upper bound on $e(c)$ for any $c \in V$, consider a furthest vertex $v \in F(c)$.
Note that, as~$x$ and~$y$ are mutually distant, $d(x,y) \geq \max\{d(x,v), d(y,v)\}$.
By Lemma~\ref{generalDualityLemma}, for every $x,y,v,c \in V$ either
$d(x,v) - d(x,y) \geq d(c,v) - d(y,c) - 2\delta$ or $d(y,v) - d(x,y) \geq d(c,v) - d(x,c) - 2\delta$ holds.
If the former is true, then $d(c,v) \leq d(x,v) - d(x,y) + d(y,c) + 2\delta \leq d(y,c) + 2\delta$.
If the latter is true, then $d(c,v) \leq d(y,v) - d(x,y) + d(x,c) + 2\delta \leq d(x,c) + 2\delta$.
Thus, $e(c) \leq \max\{d(x,c), d(y,c)\} + 2\delta$.

Moreover, if $c^*$ is a middle vertex of $I(x,y)$, i.e.,  $c^* \in S_{\lfloor d(x,y)/2 \rfloor}(x,y)\cup S_{\lfloor d(x,y)/2 \rfloor}(y,x)$, then
$e(c^*) \leq \max\{d(x,c^*), d(y,c^*)\} + 2\delta = \lceil d(x,y)/2 \rceil + 2\delta \leq \lceil 2rad(G) / 2 \rceil + 2\delta = rad(G) + 2\delta$.
In particular, since $\lceil d(x,y) / 2 \rceil \geq e(c^*) - 2\delta \geq rad(G) - 2\delta$,
 $diam(G) \geq d(x,y) \geq 2rad(G) - 4\delta - 1$.
\end{proof}

Furthermore, the eccentricity of a vertex, that is most distant from some other vertex, is close to the distance between any two mutually distant vertices. 

\begin{lemma}\label{furthestIsAlmostDiameter}
Let $G$ be a $\delta$-hyperbolic graph.
For any $x,y,c \in V$, and any furthest vertex $v \in F(c)$, $d(x,y) \leq e(v) + 2\delta$.
In particular, $e(v) \geq diam(G) - 2\delta$~\cite{Dragan2018RevisitingRD,Chepoi_2008}.
\end{lemma}
\begin{proof}
From the choice of~$v$, necessarily $d(c,v) \geq \max\{d(y,c), d(x,c\}$.
By Lemma~\ref{generalDualityLemma}, either
$d(x,v) - d(x,y) \geq d(c,v) - d(y,c) - 2\delta$ or $d(y,v) - d(x,y) \geq d(c,v) - d(x,c) - 2\delta$ holds.
If the former is true, then $d(x,y) \leq d(x,v) + d(y,c) - d(c,v) + 2\delta \leq d(x,v) + 2\delta$.
If the latter is true, then $d(x,y) \leq d(y,v) + d(x,c) - d(c,v) + 2\delta \leq d(y,v) + 2\delta$.
In either case $d(x,y) \leq \max\{d(x,v), d(y,v)\} + 2\delta \leq e(v) + 2\delta$, establishing the result. 
In particular, when $d(x,y) = diam(G)$, then $e(v) \geq diam(G) - 2\delta$.
\end{proof}

With Theorem~\ref{thm:eccOnLineBetter} and Lemma~\ref{furthestIsAlmostDiameter}, one obtains the following corollary known from~\cite{Chepoi_2008,Dragan2018RevisitingRD}.
\begin{corollary}~\cite{Chepoi_2008,Dragan2018RevisitingRD}\label{cor:beamEccToDiam}
Let $G$ be a $\delta$-hyperbolic graph.
For any vertex $x \in V$, every vertex $y \in F(x)$ has $e(y) \geq diam(G) - 2\delta \geq 2rad(G) - 6\delta - 1$.
\end{corollary}

Of a great algorithmic interest is also the following result. 
\begin{corollary}\label{eccMiddleOfBeam}
Let $G$ be a $\delta$-hyperbolic graph, where $z \in V$, $x \in F(z)$, and $y \in F(x)$.
Any vertex $c \in S_{\lfloor d(x,y)/2 \rfloor}(x,y)$ has $e(c) \leq rad(G) + 3\delta$~\cite{Chepoi_2008,Dragan2018RevisitingRD}.
Moreover, $e(c) \leq \lceil d(x,y)/2 \rceil + 4\delta$.
\end{corollary}
\begin{proof}
Let $v \in F(c)$ be a furthest vertex from $c$.
Since $y \in F(x)$, $d(x,v) \leq d(x,y)$. We also have $d(x,y) \leq 2rad(G)$.
If $d(x,c) \leq (v|y)_x$ then, by Lemma~\ref{lemDuality}, 
$e(c) = d(c,v) \leq d(x,v) - d(x,c) + 2\delta \leq d(x,y) - d(x,c) + 2\delta = \lceil d(x,y)/2 \rceil + 2\delta \leq rad(G) + 2\delta$.

On the other hand, if $d(x,c) \geq (v|y)_x$ then, by Lemma~\ref{lemDuality},  
$e(c) = d(c,v) \leq d(y,v) - d(y,c) + 2\delta$.
By Corollary~\ref{cor:beamEccToDiam}, $d(y,c) = \lceil d(x,y)/2 \rceil = \lceil e(x)/2 \rceil \geq \lceil (diam(G) - 2\delta)/2 \rceil = \lceil diam(G)/2 \rceil - \delta$.
Therefore, $e(c) \leq d(y,v) - d(y,c) + 2\delta \leq diam(G) - \lceil diam(G)/2 \rceil + 3\delta \leq rad(G) + 3\delta$.
Thus, $e(c) \leq rad(G) + 3\delta$.
Moreover, by Corollary~\ref{cor:beamEccToDiam}, $diam(G) \leq d(x,y) + 2\delta$.
Hence, $e(c) \le d(y,v) - d(y,c) + 2\delta \leq diam(G) - d(y,c) + 2\delta \leq d(x,y) - d(y,c) + 4\delta \leq \lfloor d(x,y)/2 \rfloor + 4\delta\leq \lceil d(x,y)/2 \rceil + 4\delta$. 
\end{proof}

In Section \ref{section:convexity}, we saw that the diameter in $G$ of any set $C_{\leq k}(G)$, $k\in \mathbb{N}$, is bounded by $2k + 4\delta + 1$. In particular, $diam(C(G))\le  4\delta + 1$ holds~\cite{Chepoi_2008}. Note that, in~\cite{Chepoi_2008}, it was additionally shown that all central vertices are close to a middle vertex~$c$ of a shortest $(x,y)$-path, provided that $x$ is furthest from some vertex and that $y$ is furthest from~$x$. Namely, $D(c, 5\delta + 1) \supseteq C(G)$ holds. Here, we provide such a result with respect to $C_{\leq k}(G)$ for all $k\in \mathbb{N}$.

\begin{lemma}\label{lem:centerToMiddleBeam}
Let $G$ be a $\delta$-hyperbolic graph, and let $z \in V$, $x \in F(z)$, and $y \in F(x)$.
Any middle vertex~$c \in S_{\lceil d(x,y)/2 \rceil}(x,y)\cup S_{\lceil d(x,y)/2 \rceil}(y,x)$ satisfies $D(c, 5\delta + 1 + k) \supseteq C_{\leq k}(G)$.
In particular, $D(c, 5\delta + 1) \supseteq C(G)$~\cite{Chepoi_2008}.
\end{lemma}
\begin{proof}
Consider an arbitrary vertex $u \in C_{\leq k}(G)$.
By Corollary~\ref{cor:maxMinDuality}, $d(c,u) \leq \max\{d(x,u), d(y,u)\} - \min\{d(x,c), d(y,c)\} + 2\delta$.
As $e(u) \leq rad(G) + k$, $\max\{d(x,u), d(y,u)\} \leq rad(G) + k$ holds.
As $c$ is a middle vertex of a shortest $(x,y)$-path,  $\min\{d(x,c), d(y,c)\} = \lfloor d(x,y) / 2 \rfloor$.
Since $x \in F(z)$, by Corollary~\ref{cor:beamEccToDiam}, $e(x) = d(x,y) \geq 2rad(G) - 6\delta - 1$.
Hence, $d(c,u) \leq rad(G) + k - \lfloor (2rad(G) - 6\delta - 1) / 2 \rfloor + 2\delta \leq 5\delta + 1 + k$.
\end{proof}

In~\cite{Chepoi2018FastAO}, it was proven that if $G$ is a $\tau$-thin graph, then any middle vertex~$c$ of any shortest path $P(x,y)$ between two mutually distant vertices $x$ and $y$ satisfies $D(c, k+2\tau + 1) \supseteq C_{\le k}(G)$. 
Applying the inequality~$\tau \leq 4\delta$ from Proposition~\ref{prop:thinVsHyperbolicity} to this result,
one obtains $D(c, k+8\delta + 1) \supseteq C_{\le k}(G)$.
The latter result can be improved working directly with~$\delta$. 

\begin{lemma}\label{lem:centerToMiddleLine}
Let $G$ be a $\delta$-hyperbolic graph, and let $\{x,y\}$ be a mutually distant pair.
Any middle vertex~$c \in S_{\lceil d(x,y)/2 \rceil}(x,y)\cup S_{\lceil d(x,y)/2 \rceil}(y,x)$ satisfies $D(c, 4\delta + 1 + k) \supseteq C_{\leq k}(G)$.
In particular, $D(c, 4\delta + 1) \supseteq C(G)$.
\end{lemma}
\begin{proof}
The proof is analogous to that of Lemma~\ref{lem:centerToMiddleBeam}.
However, since $x,y$ are mutually distant,  by Theorem~\ref{thm:eccOnLineBetter}, $d(x,y) \geq 2rad(G) - 4\delta - 1$.
Hence, for any $u \in C_{\leq k}(G)$,
$d(c,u) \leq rad(G) + k - \lfloor (2rad(G) - 4\delta - 1) / 2 \rfloor + 2\delta \leq 4\delta + 1 + k$.
\end{proof}

Thus, combining this with the result from~\cite{Chepoi2018FastAO}, we get  $D(c, \min\{4\delta(G), 2\tau(G)\} + 1 + k) \supseteq C_{\leq k}(G)$ for any graph $G$. 
\medskip

There are several algorithmic implications of the results of this subsection.  For an arbitrary connected graph $G=(V,E)$ and a given vertex $z\in V$, a most distant from $z$ vertex $x\in F(z)$ can be found in linear ($O(|E|)$) time by a {\em breadth-first-search} $BFS(z)$ started at $z$.
A pair of mutually distant vertices of a $\delta$-hyperbolic graph $G=(V,E)$ can be computed in $O(\delta |E|)$ total time as follows.  By Lemma~\ref{furthestIsAlmostDiameter}, if $x$ is a most distant vertex from an arbitrary vertex $z$ and $y$ is a most distant vertex from $x$, then $d(x,y)\geq diam(G)-2\delta$. Hence, using at most $O(\delta)$ {\em breadth-first-searches}, one can generate a sequence of vertices $x:=v_1,y:=v_2, v_3, \dots v_k$ with $k\leq 2\delta+2$  such that each $v_i$ is most distant from $v_{i-1}$ (with $v_0=z$) and $v_k$, $v_{k-1}$ are mutually distant vertices (the initial value $d(x,y)\geq diam(G)-2\delta$ can be improved at most $2\delta$ times).

Thus, by Theorem~\ref{thm:eccOnLineBetter},  Lemma~\ref{furthestIsAlmostDiameter},   Corollary~\ref{eccMiddleOfBeam}, Lemma~\ref{lem:centerToMiddleBeam}, and Lemma~\ref{lem:centerToMiddleLine}, 
we get the following 
additive approximations for the radius and the diameter of a $\delta$-hyperbolic graph $G$.

\begin{corollary} \label{th:rad-diam-appr}
Let $G=(V,E)$ be a $\delta$-hyperbolic graph.
\begin{enumerate}
\item[$(i)$] \cite{Chepoi_2008,Dragan2018RevisitingRD} There is a linear $(O(|E|))$ time algorithm which finds in $G$ a vertex $c$ with eccentricity at most $rad(G)+3\delta$ and a vertex $v$ with eccentricity at least $diam(G)-2\delta$. Furthermore, $C(G)\subseteq D(c,5\delta+1)$ holds.
\item[$(ii)$] There is an almost linear $(O(\delta|E|))$ time algorithm which finds in $G$ a vertex $c$ with eccentricity at most $rad(G)+2\delta$~\cite{Dragan2018RevisitingRD}. Furthermore, $C(G)\subseteq D(c,4\delta+1)$ holds.
\end{enumerate}
\end{corollary}

In a graph $G$ where vertex degrees and the hyperbolicity $\delta(G)$ are bounded by constants, the entire center $C(G)$ can be found in linear time by running a breadth-first-search from each vertex of $D(c, 5\delta+1)$ and selecting among them the vertices with smallest eccentricity. In this case, $D(c, 5\delta+1)$ contains only a constant number of vertices. 

\begin{corollary}
Let $G$ be a $\delta$-hyperbolic graph with maximum vertex degree $\Delta(G)$. 
A vertex~$c$ with $C(G) \subseteq D(c,5\delta+1)$ can be found in $O(|E|)$ time and the center $C(G)$ can be computed in $O(\Delta(G)^{5\delta+1}|E|)$ time~\cite{Chepoi_2008}.
%
A vertex~$c$ with $C(G) \subseteq D(c,4\delta+1)$ can be found in $O(\delta|E|)$ time and the center $C(G)$ can be computed in $O(\Delta(G)^{4\delta+1}|E|)$ time.
If the degrees of vertices and the hyperbolicity of $G$ are uniformly bounded by a constant, then $C(G)$ can be found in total linear time~\cite{Chepoi_2008,Dragan2018RevisitingRD}.
\end{corollary}





By Theorem~\ref{thm:eccOnLineBetter}, we get also the following left-sided additive approximations of all vertex eccentricities. 
Let  $\{x,y\}$ be a mutually distant pair of vertices of $G$. For every vertex~$v \in V$, set $\hat{e}(v) := \max\{d(x,v), d(y,v)\}$. 

\begin{corollary} \label{cor:all-ecc-left}
Let $G=(V,E)$ be a $\delta$-hyperbolic graph.  
There is an algorithm which in total almost linear $(O(\delta|E|))$ time outputs for every vertex $v\in V$ an estimate $\hat{e}(v)$ of its eccentricity $e_G(v)$ such that  $e_G(v)- 2\delta \leq \hat{e}(v) \leq {e_G}(v).$
\end{corollary}

If~the hyperbolicity $\delta(G)$ of $G$ is known in advance, we can transform $\hat{e}$ into a right-sided additive~$2\delta$-approximation by setting~$\hat{e}(v) := \max\{d(x,v), d(y,v)\} + 2\delta$. 
This approach generalizes a recently discovered eccentricity approximation for distance-hereditary graphs~\cite{ourManuscriptDHG} (the hyperbolicity of any  distance-hereditary graph is at most 1).  Unfortunately, if $\delta(G)$ is not known in advance, the best to date algorithm  for computing $\delta(G)$ has complexity $O(n^{3.69})$  and relies on some (rather impractical) matrix multiplication results~\cite{DBLP:journals/ipl/FournierIV15} (see also~\cite{DBLP:conf/compgeom/ChalopinCDDMV18} for some recent approximation algorithms). In the next subsection, we will give some right-sided additive approximations of all vertex eccentricities which do not assume any knowledge of $\delta(G)$.   

\subsection{Right-sided additive approximations of all vertex eccentricities}

In what follows, we illustrate two right-sided additive eccentricity approximations for all vertices using a notion of \emph{eccentricity approximating spanning tree} introduced in \cite{DBLP:journals/dm/Prisner00a} and investigated in~\cite{Dragan2018,Dragan2017EccentricityAT,Chepoi2018FastAO,DBLP:journals/dam/Ducoffe19a}.  We get 
a $O(|E|)$ time right-sided additive $(6\delta)$-approximations and a $O(\delta|E|)$ time right-sided additive $(4\delta+1)$-approximations. 

A spanning tree $T$ of a graph $G$ is called an {\em eccentricity $k$-approximating spanning tree} if for every vertex $v$ of $G$  $e_T(v)\leq e_G(v)+ k$ holds~\cite{DBLP:journals/dm/Prisner00a}. All  $(\alpha_1, \triangle)$-metric graphs (including chordal graphs and the underlying graphs of 7-systolic complexes)  admit eccentricity 2-approximating spanning trees~\cite{Dragan2017EccentricityAT}. 
An eccentricity 2-approximating spanning tree of a chordal graph can be computed in linear time~\cite{Dragan2018}. An eccentricity $k$-approximating spanning tree with minimum $k$ can be found in $O(|V||E|)$ time for any graph $G$~\cite{DBLP:journals/dam/Ducoffe19a}.
It is also known~\cite{Chepoi2018FastAO} that if~$G$ is a $\tau$-thin graph, then $G$ admits an eccentricity $(2\tau)$-approximating spanning tree constructible in $O(\tau|E|)$ time and an eccentricity $(6\tau+1)$-approximating spanning tree constructible in $O(|E|)$ time. 
Applying the inequality $\tau \leq 4\delta$ from Proposition~\ref{prop:thinVsHyperbolicity}, we get that every $\delta$-hyperbolic graph admits an eccentricity $8\delta$-approximating spanning tree constructible in $O(\delta|E|)$ time and an eccentricity $(24\delta+1)$-approximating spanning tree constructible in $O(|E|)$ time. Both these results can be significantly improved working directly with~$\delta$.

\begin{theorem}\label{thm:eccApproxTreeByLine}
Let $G$ be a $\delta$-hyperbolic graph.
If $c$ is a middle vertex of any shortest $(x,y)$-path between a pair $\{x,y\}$ of mutually distant vertices of $G$ and $T$ is a $BFS(c)$-tree of $G$, then, for every vertex $v$ of $G$, $e_G(v)\leq e_T(v)\leq e_G(v)+ 4\delta+1.$ That is, $G$ admits an eccentricity $(4\delta+1)$-approximating spanning tree constructible in $O(\delta|E|)$ time. 
\end{theorem}
\begin{proof}
%
The eccentricity in~$T$ of any vertex~$v$ can only increase compared to its eccentricity in~$G$. Hence, $e_G(v) \leq e_T(v)$.
By the triangle inequality and the fact that all distances from vertex $c$ are preserved in $T$, 
$e_T(v) \leq d_T(v,c) + e_T(c) = d_G(v,c) + e_G(c)$.
We know that $e_G(v) \geq \max\{d_G(y,v), d_G(x,v)\}$. By 
Corollary~\ref{cor:maxMinDuality}, also $d_G(v,c) - \max\{d_G(y,v),\ d_G(x,v) \} \leq 2\delta - \min\{d_G(y,c),\ d_G(x,c)\}$ holds. Since $c$ is a middle vertex of a shortest $(x,y)$-path,
necessarily $\min\{d_G(y,c),\ d_G(x,c)\} \geq \lfloor d_G(x,y)/2 \rfloor$ and,
by Theorem~\ref{thm:eccOnLineBetter}, $e_G(c) \leq \lceil d_G(x,y)/2 \rceil + 2\delta$.
Combining all these, we get
\begin{align*}
  e_T(v) - e_G(v) & \leq d_G(v,c) + e_G(c) - e_G(v) \\
  				  & \leq d_G(v,c) - \max\{d_G(y,v), d_G(x,v)\} + e_G(c) \\
  				  & \leq 2\delta - \min\{d_G(y,c),\ d_G(x,c)\} + e_G(c) \\
  				  & \leq 2\delta - \lfloor d_G(x,y) / 2 \rfloor + e_G(c) \\
  				  & \leq 2\delta - \lfloor d_G(x,y) / 2 \rfloor + \lceil d_G(x,y)/2 \rceil + 2\delta \\
  				  & \leq 4\delta + 1.
\end{align*}
\end{proof}

We next give a $O(|E|)$ time right-sided additive eccentricity $(6\delta + 1 - k)$-approximation for any constant integer $k$, $0\le k\leq 2\delta$.
\begin{theorem}\label{thm:eccApproxTreeByManyBeams}
Let $G$ be a $\delta$-hyperbolic graph and $k$ be an integer from $[0,2\delta]$. Let $u_0,u_1,\dots,u_{k+2}$ be a sequence of vertices of $G$ such that $u_0$ is an arbitrary start vertex and each $u_{i+1}$ is a vertex furthest from $u_i$ $(0\le i\le k+1)$. 
If $c$ is a middle vertex of any shortest $(u_{k+1},u_{k+2})$-path and $T$ is a $BFS(c)$-tree of $G$, then, for every vertex $v$ of $G$, $e_G(v)\leq e_T(v)\leq e_G(v)+ 6\delta+1-k.$ That is, $G$ admits an eccentricity $(6\delta+1-k)$-approximating spanning tree constructible in $O(k|E|)$ time. 
%
%
%
\end{theorem}
\begin{proof}
Recall that if $\{u_{k+1}, u_{k+2}\}$ or an earlier pair  $\{u_{i+1}, u_{i+2}\}$ ($i<k$) is a mutually distant pair then, by Theorem~\ref{thm:eccApproxTreeByLine}, $T$ is an eccentricity $(4\delta+1)$-approximating spanning tree. 
Therefore, in what remains, we assume that $d_G(u_{i+1}, u_{i+2}) \ge d_G(u_{i}, u_{i+1})+1$ for all $i \leq k$.
Hence, by Corollary~\ref{cor:beamEccToDiam}, $d_G(u_{k+1}, u_{k+2}) \ge d_G(u_1,u_2) + k \geq diam(G) - 2\delta + k$.

We first claim that 
a middle vertex $c$ of any shortest $(u_{k+1}, u_{k+2})$-path satisfies $e_G(c) \leq \lceil d_G(u_{k+1}, u_{k+2})/2 \rceil + 4\delta - k$. 
Let $t \in F(c)$, i.e., $e_G(c) = d_G(c,t)$.
By Lemma~\ref{lemDuality}, either $d_G(c,t) \leq d_G(u_{k+1},t) - d_G(u_{k+1},c) + 2\delta$ or $d_G(c,t) \leq d_G(u_{k+2},t) - d_G(u_{k+2},c) + 2\delta$. 
If the former is true then, since $d_G(u_{k+1},t) \leq d_G(u_{k+1},u_{k+2})$, we have 
$e_G(c)\le d_G(u_{k+1},t) - d_G(u_{k+1},c) + 2\delta \leq d_G(u_{k+1},u_{k+2}) - d(u_{k+1},c)  + 2\delta \leq \lceil d_G(u_{k+1},u_{k+2})/2 \rceil + 2\delta$.
If the latter is true then, since $d_G(u_{k+1}, u_{k+2}) \geq diam(G) - 2\delta + k \geq d_G(u_{k+2},t) - 2\delta + k$, we get 
 $e_G(c)\le d_G(u_{k+2},t) - d_G(u_{k+2},c) + 2\delta \leq 
d_G(u_{k+1}, u_{k+2}) + 2\delta - k - d_G(u_{k+2},c) + 2\delta \leq \lceil d_G(u_{k+1},u_{k+2}) / 2 \rceil + 4\delta - k$.
As $k \leq 2\delta$, in either case,  $e_G(c) \leq \lceil d_G(u_{k+1},u_{k+2}) / 2 \rceil + 4\delta - k$, establishing the claim.

Set now $x := u_{k+1}$ and $y := u_{k+2}$.
The remainder of the proof follows the proof of Theorem~\ref{thm:eccApproxTreeByLine} with one adjustment:
replace the application of Theorem~\ref{thm:eccOnLineBetter} which yields $e_G(c) \leq \lceil d_G(x,y)/2 \rceil + 2\delta$
with our claim which yields $e_G(c) \leq \lceil d_G(x,y)/2 \rceil + 4\delta - k$.
Hence, $e_T(v) - e_G(v) \leq 2\delta - \lfloor d_G(x,y) / 2 \rfloor + \lceil d_G(x,y)/2 \rceil + 4\delta - k \leq 6\delta + 1 - k$.
\end{proof}

Theorem~\ref{thm:eccApproxTreeByManyBeams} generalizes Theorem~\ref{thm:eccApproxTreeByLine} (when $k=2\delta$, we obtain Theorem~\ref{thm:eccApproxTreeByLine}). 
Also, when $k=1$, we get an eccentricity $(6\delta)$-approximating spanning tree constructible in $O(|E|)$ time. 

Note that the eccentricities of all vertices in any tree $T=(V,U)$ can be computed in $O(|V|)$ total time. It is a folklore by now that for trees the following facts are true:
(1) The center $C(T)$ of any tree $T$ consists of one vertex or two adjacent vertices; (2) The center $C(T)$ and the radius $rad(T)$ of any tree $T$ can be found in linear time; (3) For every vertex $v\in V$, $e_T(v)=d_T(v,C(T))+rad(T)$.
Hence, using $BFS(C(T))$ on $T$ one can compute $d_T(v,C(T))$ for all $v\in V$ in total $O(|V|)$ time. Adding now $rad(T)$ to $d_T(v,C(T))$, one gets  $e_T(v)$ for all $v\in V$. Consequently, by Theorem~\ref{thm:eccApproxTreeByLine}  and Theorem~\ref{thm:eccApproxTreeByManyBeams}, we get the following additive approximations for the vertex eccentricities in $\delta$-hyperbolic graphs.  

\begin{corollary} \label{th:all-ecc--appr}
Let $G=(V,E)$ be a $\delta$-hyperbolic graph.
\begin{enumerate}
\item[$(i)$] There is an algorithm which in total linear $(O(|E|))$ time outputs for every vertex $v\in V$ an estimate $\hat{e}(v)$ of its eccentricity $e_G(v)$ such that $e_G(v)\leq \hat{e}(v)\leq e_G(v)+ 6\delta.$
\item[$(ii)$] There is an algorithm which in total almost linear $(O(\delta|E|))$ time outputs for every vertex $v\in V$ an estimate $\hat{e}(v)$ of its eccentricity $e_G(v)$ such that $e_G(v)\leq \hat{e}(v)\leq e_G(v)+ 4\delta+1.$
\end{enumerate}
\end{corollary}

As $\delta(G)\le \tau(G)$ for any graph $G$, Corollary~\ref{th:all-ecc--appr}$(i)$ improves the corresponding result from~\cite{Chepoi2018FastAO}. Also, combining Corollary~\ref{th:all-ecc--appr}$(ii)$ with the corresponding  result from~\cite{Chepoi2018FastAO}, we get that, for any graph $G$, 
there is an algorithm which in total almost linear $(O(\delta|E|))$ time outputs for every vertex $v\in V$ an estimate $\hat{e}(v)$ of its eccentricity $e_G(v)$ such that $e_G(v)\leq \hat{e}(v)\leq e_G(v)+ \min\{4\delta(G)+1, 2\tau(G)\}.$

\bibliographystyle{plain}
\bibliography{bibliography}

\begin{thebibliography}{10}

\bibitem{10.5555/2884435.2884463}
Amir Abboud, Virginia~Vassilevska Williams, and Joshua Wang.
\newblock Approximation and fixed parameter subquadratic algorithms for radius
  and diameter in sparse graphs.
\newblock In {\em SODA 2016}, pages 377--391.

\bibitem{AADr}
Muad Abu-Ata and Feodor~F Dragan.
\newblock Metric tree-like structures in real-world networks: an empirical
  study.
\newblock {\em Networks}, 67(1):49--68, 2016.

\bibitem{AdcockSM13}
Aaron~B. {Adcock}, Blair~D. {Sullivan}, and Michael~W. {Mahoney}.
\newblock Tree-like structure in large social and information networks.
\newblock In {\em ICDM 2013}, pages 1--10.

\bibitem{Alonso_1991aa}
J.~M. Alonso, T.~Brady, D.~Cooper, V.~Ferlini, M.~Lustig, M.~Mihalik,
  M.~Shapiro, and H.~Short.
\newblock {\em Notes on word hyperbolic groups}.
\newblock World Scientific, 1991.

\bibitem{Alrasheed_2016}
Hend Alrasheed and Feodor~F Dragan.
\newblock Core--periphery models for graphs based on their
  $\delta$-hyperbolicity: An example using biological networks.
\newblock {\em Journal of Algorithms \& Computational Technology},
  11(1):40--57, 2016.

\bibitem{Bo++}
Michele Borassi, David Coudert, Pierluigi Crescenzi, and Andrea Marino.
\newblock On computing the hyperbolicity of real-world graphs.
\newblock In {\em Algorithms-ESA 2015}, pages 215--226.

\bibitem{Bridson_1999}
Martin~R. Bridson and Andr{\'e} Haefliger.
\newblock Metric spaces of non-positive curvature.
\newblock {\em Grundlehren der mathematischen Wissenschaften}, 1999.

\bibitem{DBLP:conf/compgeom/ChalopinCDDMV18}
J{\'{e}}r{\'{e}}mie Chalopin, Victor Chepoi, Feodor~F. Dragan, Guillaume
  Ducoffe, Abdulhakeem Mohammed, and Yann Vax{\`{e}}s.
\newblock Fast approximation and exact computation of negative curvature
  parameters of graphs.
\newblock In {\em SoCG 2018}, pages 22:1--22:15 (to appear in \emph{Discrete
  Comp. Geometry}).

\bibitem{Chepoi-center-triangulated}
Victor Chepoi.
\newblock Centers of triangulated graphs.
\newblock {\em Math. Notes}, 43:143--151, 1988.

\bibitem{Chepoi_2008}
Victor Chepoi, Feodor Dragan, Bertrand Estellon, Michel Habib, and Yann
  Vax\`{e}s.
\newblock Diameters, centers, and approximating trees of $\delta$-hyperbolic
  geodesic spaces and graphs.
\newblock In {\em SoCG 2008}, pages 59--68.

\bibitem{Chepoi2018FastAO}
Victor Chepoi, Feodor~F. Dragan, Michel Habib, Yann Vax{\`e}s, and Hend
  Alrasheed.
\newblock Fast approximation of eccentricities and distances in hyperbolic
  graphs.
\newblock {\em Journal of Graph Algorithms and Applications}, 23(2):393--433,
  2019.

\bibitem{Chepoi:2017:CCI:3039686.3039835}
Victor Chepoi, Feodor~F. Dragan, and Yann Vax\`{e}s.
\newblock Core congestion is inherent in hyperbolic networks.
\newblock In {\em SODA 2017}, pages 2264--2279.

\bibitem{FDraganPhD}
Feodor~F. Dragan.
\newblock {\em Centers of Graphs and the Helly Property (in Russian)}.
\newblock PhD thesis, Moldava State University, Chi{\c s}in{\u a}u, 1989.

\bibitem{Dragan2018}
Feodor~F. Dragan.
\newblock An eccentricity 2-approximating spanning tree of a chordal graph is
  computable in linear time.
\newblock {\em Information Processing Letters}, 154:105873, 2020.

\bibitem{ourManuscriptDHG}
Feodor~F. Dragan and Heather~M. Guarnera.
\newblock Eccentricity function in distance-hereditary graphs.
\newblock {\em CoRR}, abs/1907.05445, 2019.

\bibitem{Dragan2018RevisitingRD}
Feodor~F. Dragan, Michel Habib, and Laurent Viennot.
\newblock Revisiting radius, diameter, and all eccentricity computation in
  graphs through certificates.
\newblock {\em CoRR}, abs/1803.04660, 2018.

\bibitem{Dragan2017EccentricityAT}
Feodor~F. Dragan, Ekkehard K{\"o}hler, and Hend Alrasheed.
\newblock Eccentricity approximating trees.
\newblock {\em Discrete Applied Mathematics}, 232:142--156, 2017.

\bibitem{DBLP:journals/dam/Ducoffe19a}
Guillaume Ducoffe.
\newblock Easy computation of eccentricity approximating trees.
\newblock {\em Discrete Applied Mathematics}, 260:267--271, 2019.

\bibitem{DBLP:journals/ipl/FournierIV15}
Herv{\'{e}} Fournier, Anas Ismail, and Antoine Vigneron.
\newblock Computing the gromov hyperbolicity of a discrete metric space.
\newblock {\em Information Processing Letters}, 115(6-8):576--579, 2015.

\bibitem{gromov90}
E.~Ghys and P.~de~la Harpe~eds.
\newblock {Sur les Groupes Hyperboliques d'apr{\`e}s Mikhael Gromov}.
\newblock {\em Progress in Mathematics}, 1990.

\bibitem{Gromov1987}
M.~Gromov.
\newblock {\em Hyperbolic Groups}, pages 75--263.
\newblock Springer, New York, NY, 1987.

\bibitem{KeSN16}
W~Sean Kennedy, Iraj Saniee, and Onuttom Narayan.
\newblock On the hyperbolicity of large-scale networks and its estimation.
\newblock In {\em Big Data 2016}, pages 3344--3351.

\bibitem{Brandes}
Dirk Kosch{\"u}tzki, Katharina~Anna Lehmann, Leon Peeters, Stefan Richter,
  Dagmar Tenfelde-Podehl, and Oliver Zlotowski.
\newblock Centrality indices.
\newblock In {\em Network Analysis}, pages 16--61. Springer, 2005.

\bibitem{NaSa}
Onuttom Narayan and Iraj Saniee.
\newblock Large-scale curvature of networks.
\newblock {\em Physical Review E}, 84(6):066108, 2011.

\bibitem{DBLP:journals/dm/Prisner00a}
Erich Prisner.
\newblock Eccentricity-approximating trees in chordal graphs.
\newblock {\em Discrete Mathematics}, 220(1-3):263--269, 2000.

\bibitem{Roditty_2013}
Liam Roditty and Virginia Vassilevska~Williams.
\newblock Fast approximation algorithms for the diameter and radius of sparse
  graphs.
\newblock In {\em STOC 2013}, pages 515--524.

\bibitem{ShavittT08}
Yuval Shavitt and Tomer Tankel.
\newblock Hyperbolic embedding of internet graph for distance estimation and
  overlay construction.
\newblock {\em IEEE/ACM Transactions on Networking}, 16(1):25--36, 2008.

\end{thebibliography}

\section*{Appendix}

\textbf{List of notations}
\begin{table}[ht]
    \begin{tabular}{cp{0.8\textwidth}}
        $G=(V,E)$   & a graph \\
        $\langle S \rangle$ & subgraph induced by $S \subseteq V$ \\
        $d(u,v)$  & distance from $u$ to~$v$ \\
        $d(u,S)$  & smallest distance from $u$ to a vertex of set~$S$ \\ 
        $e(v)$    & eccentricity of~$v$ \\ 
        $I(u,v)$    & interval (set of vertices on a shortest path between $u$ and~$v$) \\
        $S_k(u,v)$  & interval slice (set of vertices on $I(u,v)$ that are at distance $k$ from vertex~$u$) \\
        $D(S,r)$    & disk (set of vertices with distance at most $r$ to $S$) \\
        $F(v)$      & set of vertices furthest from~$v$ \\
        $(x|y)_z$   & Gromov product of $x,y$ with respect to~$z$ \\
        $rad(G)$    & radius  \\ 
        $diam(G)$   & diameter \\ 
        $\delta(G)$ & hyperbolicity \\
        $\tau(G)$   & thinness \\
        $\Delta(G)$ & maximum vertex degree \\ 
        $C(G)$      & center \\
        $C_{\le k}(G)$ & set of vertices with eccentricity at most $rad(G)+k$ \\
        $C_{k}(G)$ & set of vertices with eccentricity equal to $rad(G)+k$ \\
        $P(y,x)$ & shortest path from $y$ to $x$ \\
        $\mathcal{U}(P(y,x))$ & number of up-edges on $P(y,x)$ \\
        $\mathcal{H}(P(y,x))$ & number of horizontal-edges on $P(y,x)$ \\
        $\mathcal{D}(P(y,x))$ & number of down-edges on $P(y,x)$
    \end{tabular}
\end{table}

\end{document}